\documentclass[12pt,reqno,a4paper]{amsart}
\usepackage{amsmath,amssymb,amsfonts,amscd,enumerate,tikz}
\usepackage[mathscr]{eucal}
\usepackage{color}
\usetikzlibrary{positioning} 
\usepackage{chngcntr}
\usetikzlibrary{decorations.pathmorphing}
\usetikzlibrary{arrows.meta}

\numberwithin{equation}{section}

\usepackage{hyperref}  

\date{}

\title{A quantum anchor for higher Koszul brackets}

\author{Ekaterina Shemyakova and Yagmur Yilmaz}
\address{{University of Toledo, USA}}
\email{Ekaterina.Shemyakova@utoledo.edu}
\thanks{This work was partially supported by Simons Foundation collaboration grant for 
mathematicians number 846970}

\def\co{\colon\thinspace}

\newtheorem{theorem}{Theorem}[section]
\newtheorem{proposition}{Proposition}[section]
\newtheorem{lemma}{Lemma}[section]

\theoremstyle{definition}
\newtheorem{definition}{Definition}[section]
\newtheorem{example}{Example}[section]
\newtheorem{remark}{Remark}[section]

\def\co{\colon\thinspace}

\renewcommand{\leq}{\leqslant}
\renewcommand{\geq}{\geqslant}

\newcommand{\ih}{\frac{i}{\hbar}}
\newcommand{\hp}{{\hat p}}
\DeclareMathOperator{\symb}{\sigma}
\DeclareMathOperator{\Dens}{\mathfrak{Dens}}
\newcommand{\HSch}{H_{\text{Sch}}}
\newcommand{\Hd}{H_{\text{deRham}}}
\newcommand{\HdP}{H_{d_P}}

\renewcommand{\div}{{\mathop{\mathrm{div}}}}

\newcommand{\Dbar}{%
   {{D\mkern-14 mu
   \mathchoice{\raisebox{-2pt}{$\displaystyle \mathchar'26$}}
             {\raisebox{-2pt}{$\mathchar'26$}}
             {\raisebox{-1pt}{$\scriptstyle \mathchar'26$}}
             {\raisebox{-0.5pt}{$\scriptscriptstyle \mathchar'26$}}}
             \mkern 5 mu}}

\DeclareMathOperator{\fun}{\mathit{C^{\infty}}}

\DeclareMathOperator{\Vol}{Vol}

\DeclareMathOperator{\ord}{ord}

\DeclareMathOperator{\sign}{sgn}

\DeclareMathOperator{\MX}{MX}

\renewcommand{\P}{\Pi}

\newcommand{\der}[2]{{\frac{\partial {#1}}{\partial {#2}}}}

\newcommand{\Z}{{\mathbb Z_{2}}}
\newcommand{\ZZ}{{\mathbb Z}}

\newcommand{\p}{\partial}

\newcommand{\e}{\varepsilon}

\renewcommand{\O}{\Omega}
\newcommand{\D}{\Delta}
\renewcommand{\o}{\omega}

\newcommand{\la}{{\lambda}}

\renewcommand{\d}{\delta}

\newcommand{\ft}{{\tilde f}}

\newcommand{\at}{{\tilde a}}
\newcommand{\bt}{{\tilde b}}

\newcommand{\At}{{\tilde A}}
\newcommand{\Bt}{{\tilde B}}

\newcommand{\Lt}{{\tilde L}}

\newcommand{\brh}{{\boldsymbol{\rho}}}

\newcommand{\bof}{{\boldsymbol{f}}}
\newcommand{\bg}{{\boldsymbol{g}}}

\DeclareMathOperator{\Vect}{\mathrm{Vect}}

\DeclareMathOperator{\Mult}{\mathfrak{A}}

\newcommand{\lsch}{{[\![}}
\newcommand{\rsch}{{]\!]}}

\newcommand{\Sinf}{S_{\infty}}
\newcommand{\Pinf}{P_{\infty}}
\newcommand{\Linf}{L_{\infty}}

\newcommand{\void}{\varnothing}

\newcommand{\tto}{{\linethickness{3pt}
		  \,\begin{picture}(1,0)
                   \put(0,0.26){\line(1,0){0.95}}
                   \put(0,0){$\boldsymbol{\rightarrow}$}
                  \end{picture}
                  }\,
}

\newcommand{\oto}{{\linethickness{0.5pt}
		  \,\begin{picture}(1,0)
		  \put(0.07,0.175){\line(0,1){0.2}}
                   \put(-0.01,0){$\boldsymbol{\Rightarrow}$}
                  \end{picture}
                  }\,
}

\newcommand{\du}{\text{dual}}

\newcommand{\sta}{\star}

\begin{document}
\maketitle

\begin{abstract}
It is well known that the chain map between the de Rham and 
Poisson complexes on a Poisson manifold also maps the 
Koszul bracket of differential forms into the Schouten 
bracket of multivector fields.

In the generalized case of a $P_\infty$-structure,  where 
a Poisson bivector $P$ is replaced by an arbitrary even
multivector obeying $\lsch P,P\rsch=0$,   an analog  of the chain map and an
$L_\infty$-morphism from the higher Koszul brackets into 
the Schouten bracket are also known; however, they differ 
significantly in nature.

In the present paper, we address the problem of quantizing 
this picture. In particular, we show that the 
$L_\infty$-morphism is quantized into a single linear 
operator, which is a formal Fourier integral 
operator.

This paper employs Voronov's thick 
morphism technique and quantum Mackenzie-Xu 
transformations in the framework of  
$L_\infty$-algebroids.
\end{abstract}

\tableofcontents
%
%
%
%

\section{Introduction}
Nowadays there is great interest in  quantization  of various geometric structures in 
some precise sense, see e.g. Behrend-Peddie-Xu~\cite{behrend-peddie-xu:2023} and 
Pridham~\cite{pridham:outline2018, pridham:deform2019}. 

For a Lie algebroid $E\to M$, one may see as its quantization, finding a second-order 
odd differential operator (a BV operator) generating the corresponding odd bracket, 
i.e. the Lie-Schouten bracket on functions on $\Pi E^*$. For $L_{\infty}$-algebroids, 
the analog of that is more delicate and requires using $\hbar$-differential 
operators. (We explain it below.)

In our previous work~\cite{shemy:koszul}, we studied  an 
$\hbar$-differential   operator $\Delta_P$ that generates     higher Koszul brackets 
introduced by  Khudaverdian-Voronov in~\cite{tv:higherpoisson}. It can be seen as 
giving a quantization of the $\Linf$-algebroid structure  of the cotangent   algebroid 
of an $\Pinf$-manifold $(M,P)$. (Here $P\in \fun(\Pi T^*M)$ is even and  $\lsch 
P,P\rsch=0$.)

In the present paper, we are concerned with constructing an integral operator that intertwines the operator $\Delta_P$ with the divergence operator $-\hbar^2\delta_{\brh}$ generating the canonical Schouten bracket. It can be interpreted as giving a quantization of the anchor in the cotangent $\Linf$-algebroid.

The paper is organized as follows.  Sec.~\ref{sec:Preliminaries} covers the 
preliminaries. We stress the convenience of 
$L_\infty$- and Lie algebroid setting, in their different manifestations.  All 
objects of interest are elements within this framework. We recall  the 
$Q$-manifolds language, which  is indispensible for working with higher structures. 
We also review thick morphisms and the algebraic theory of $\hbar$-differential 
operators. 
In Sec.~\ref{sec:statement_of_the_problem}, we present the problem.
This includes details (omitted in the original paper~\cite{tv:highkosz})
of constructing an $L_\infty$-morphism
between Higher Koszul and Schouten brackets using thick morphisms. 
In Sec.~\ref{sec:solution}, we present the solution.

\section{Preliminaries}
\label{sec:Preliminaries}
\subsection{Notations}
\label{subsec:Notations}
For an ordinary manifold $M$, smooth functions on $\Pi T^*M$ can be identified 
with inhomogeneous multivector fields.
Here $\Pi$ is the parity reversion functor, which formally reverses the parity of 
the fiber coordinates of $T^*M$ while leaving the transformation functions unchanged. 
Similarly, smooth functions on $\Pi TM$ can be 
identified with inhomogeneous differential forms. In the case of a supermanifold,
some of the fiber coordinates may be even and 
so smooth functions on $\Pi TM$ will be pseudo-differential forms (see e.g. 
Manin~\cite{manin:gaugeeng}), which we will refer to as differential forms for 
brevity.
The same applies to $\Pi 
T^*M$ and multivector fields. We denote:
\begin{equation*}
  \Mult(M) = C^\infty(\Pi T^*M) \quad \text{and} \quad \O(M) = C^\infty(\Pi TM) \, .
\end{equation*}
Given $M$ with local coordinates
 $x^a$, where the parities are denoted as $\widetilde{x^a}=\widetilde{a}$, we use the 
following 
notation for coordinates on  the  tangent and cotangent bundles:  for $T^*M$,   the 
induced coordinates on the fibers $p_a$, $\widetilde{p_a}=\widetilde{a}$; for $\Pi 
TM$,  the induced coordinates on the fibers $dx^a$ (we consider the odd version of  
the de Rham differential $d$) with $\widetilde{dx^a}=\widetilde{a}+1$; and for
$\Pi T^*M$,    the induced coordinates on the fibers $x^*_a$ with 
 $\widetilde{x^*_a}=\widetilde{a}+1$.

\subsection{Odd and even Poisson brackets}
\label{subsec:Poissonbrackets}
\begin{definition}
An \emph{even Poisson algebra} is a commutative 
superalgebra $A$ with an even  bracket satisfying
\begin{align*} 
 \{a,b\}&=-(-1)^{\widetilde{a} \widetilde{b}}\{b,a\} \, ,\\
   \{a,\{b,c\}\} &
    =\{\{a,b\},c\}+(-1)^{\widetilde{a}\widetilde{b}}\{b,\{a,c\}\}  \, , \\
\{a, bc\} &= \{a, b\}c + (-1)^{\widetilde{a} \widetilde{b}} b \{a, c\} \, .
\end{align*} 
\end{definition}
This bracket is linear without signs, meaning $\{ ka,b\}=k\{a,b\}$, 
    $\{a,bk\}=\{a,b\}k$, and 
    $\{ak,b\}=\{a,kb\}$.   

\begin{definition} 
An \emph{odd Poisson algebra in its antisymmetric version} (also known as a Schouten 
or Gerstenhaber algebra, and the bracket is sometimes referred to as an antibracket) 
is a commutative superalgebra 
$A$ equipped with an odd bracket that satisfies the following properties
\begin{align*}
     \{a,b\}&=-(-1)^{(\widetilde{a}+1)( \widetilde{b}+1)}\{b,a\}   \, , \\
       \{a,\{b,c\}\}
        &=\{\{a,b\},c\}+(-1)^{(\widetilde{a}+1)( \widetilde{b}+1)}\{b,\{a,c\}\}  \, , 
\\
    \{a, bc\} &= \{a, b\}c + (-1)^{(\widetilde{a}+1) \widetilde{b}} b \{a, c\} \, .
\end{align*}
\end{definition}

The bracket is linear:
    $\{ ka,b\}=k\{a,b\}$, 
    $\{a,bk\}=\{a,b\}k$, and 
    $\{ak,b\}=(-1)^{\widetilde{k}}\{a,kb\}$.    
So, mnemonically, the sign is placed on the comma.

\begin{definition}
An  \emph{odd Poisson algebra in its symmetric version} is a commutative superalgebra 
$A$ with an odd bracket satisfying 
\begin{align*}
\{a,b\} &=(-1)^{\widetilde{a} \widetilde{b}}\{b,a\} \, \\
\{a,\{b,c\}\}       
&=(-1)^{\widetilde{a}+1}
\{ \{a,b\},c\}+
(-1)^{(\widetilde{a}+1)(\widetilde{b}+1)}\{b,\{a,c\}
\}  \, , \\
   \{a, bc\} &= \{a, b\}c + (-1)^{(\widetilde{a}+1) \widetilde{b}} b \{a, c\} \, .
\end{align*}
\end{definition}
 Here the sign is on the opening bracket. 
 The two versions of odd bracket can be converted into one  another 
  using
   $\{a,b\}_{\text{sym}}=(-1)^{\widetilde{a}}\{a,b\}_{\text{anti}}$ or 
$\{a,b\}_{\text{sym}}=(-1)^{\widetilde{a}+1}\{a,b\}_{\text{anti}}$.
 
\begin{example} 
\label{ex:evenPoisson_of_functions}
An even Poisson bracket on $\fun (M)$ is defined by an even Poisson 
bivector 
\begin{equation}
\label{eq:bivector}
P =\frac{1}{2} 
P^{ab} x_b^* x_a^* \in 
 C^\infty(\Pi  T^*M) \, , \  P^{ab} = (-1)^{(\widetilde{a}+1)(\widetilde{b}+1)} 
P^{ba} \, ,  \
\lsch P,P \rsch=0 \, , \ \widetilde{P}=0 \, .
\end{equation}
Here $\lsch -,- \rsch$ denotes the canonical Schouten bracket which we introduce 
below.
The bracket is defined by:
\begin{equation}\label{eq:Poisson_functions}
    \{ f,g\} = -(-1)^{\widetilde{a}(\widetilde{f}+1)} P^{ab} \frac{\partial f}{ 
\partial x^b}
    \frac{\partial g}{ \partial x^a} \, .
\end{equation}
\end{example}

\begin{example} 
\label{ex:can_Poisson_br_of_Hamiltonians}
The canonical even Poisson bracket on   
$C^\infty(T^*M)$ (such functions are known as Hamiltonians) is an 
even Poisson bracket defined by 
\begin{equation*}
\{f,g\}=0 \, , \  \{H_X,f\}=X(f) \, , \{H_X,H_Y\}=H_{[X,Y]} \, , 
\end{equation*}
where $H_X=X^ap_a$, for $X=X^a\frac{\p}{\p x^a}$. This definition is then extended 
to other Hamiltonians by the Leibniz rule.
  In particular, $\{p_a , x^b\} = \delta_a^b=-(-1)^{\widetilde{a}}(x^b,p_a\}$.  In 
local 
coordinates,
\begin{equation}
  \label{eq:can_Poiss}
      \{ H \, , \, G \} =  (-1)^{\widetilde{a}(\widetilde{H}+1)}  \frac{\partial H}{ 
\partial p_a}
    \frac{\partial G}{ \partial x^a} -(-1)^{\widetilde{a} \widetilde{H}}  
\frac{\partial H}{ \partial x^a}
    \frac{\partial G}{ \partial p_a}  \, .
  \end{equation}
\end{example}

\begin{example} 
\label{ex:Schouten_bracket}
The Schouten bracket or the canonical odd Poisson bracket of 
multivector fields is an odd bracket on $\fun(\Pi T^*M)$ defined as follows. We will 
be using the symmetric version. For
 $f,g \in C^\infty(M)$ and vector fields 
$X, Y$, 
it is defined by 
 \begin{equation}
 \lsch f,g\rsch=0 \, , \  \lsch P_X,f\rsch= X(f)  \, , \ \lsch P_X,P_Y\rsch  = 
(-1)^{\widetilde{X}}P_{[X,Y]} \, ,
 \end{equation}
 and then can be extended by the Leibniz rule to other cases. 
  In particular, $\lsch x^*_a , x^b \rsch = (-1)^{\widetilde{a}}\delta_a^b=\lsch 
x^b,x^*_a\rsch$.
 Here 
$P_X=(-1)^{\widetilde{X}}X^ax^*_a$ for $X=X^a\frac{\p}{\p x^a}$. In local 
coordinates, the bracket is given by: 
 \begin{equation}
 \label{eq:can_Schouten}
      \lsch F \, , \, G\rsch  =  (-1)^{\widetilde{a}(\widetilde{F}+1)}  
\left(\frac{\partial F}{ 
\partial x^*_a}
    \frac{\partial G}{ \partial x^a} +(-1)^{\widetilde{F}}  
\frac{\partial F}{ \partial x^a}
    \frac{\partial G}{ \partial x^*_a}\right)  \, .
  \end{equation}
\end{example}
   
\begin{example}
\label{ex:Koszul}
The Koszul bracket is an odd Poisson bracket defined on 
differential forms on a Poisson manifold. For a manifold $M$ with a Poisson bivector 
$P$ as in~\eqref{eq:bivector}, we set (in the symmetric version):  
\begin{equation} \label{eq:Koszul_form2}
   [f,g]_P = 0 \,   , \
 [f,dg]_P =  (-1)^{\widetilde{f}} \{ f,g\}_P \,   ,   \ [ df, dg]_P = - 
(-1)^{\widetilde{f}}   d \{ f,g\}_P \, .
\end{equation}
In particular, we have
\begin{equation}
 [x^a, x^b]_P = 0, \quad [x^a, dx^b]_P = -P^{ab}, \quad [dx^a, dx^b]_P = dP^{ab}.
\end{equation}
This definition extends to other forms by the Leibniz rule.
\end{example}
 
One can show that an even Poisson bracket on $M$ 
(Ex.~\ref{ex:evenPoisson_of_functions}) 
can be obtained from the Schouten bracket of multivector fields
(Ex.~\ref{ex:Schouten_bracket}) 
as a ``derived bracket''. Given an even bivector $P$ as 
in~\eqref{eq:bivector}, we have: 
\begin{equation}
\{f,g\}_P = 
\lsch f,\lsch P,g \rsch \rsch \, .
\end{equation}

Similarly, an odd Poisson bracket on $M$ 
can be obtained as a derived bracket from 
the canonical bracket of Hamiltonians (Ex.~\ref{ex:can_Poisson_br_of_Hamiltonians}).
Given an odd fiberwise quadratic Hamiltonian 
\begin{equation} \label{eq:odd_Hamiltonian_quadratic}
H = \frac{1}{2} H^{ab} p_b p_a \in 
\fun(T^*M) \, , \ 
H^{ab}=(-1)^{\widetilde{a}\widetilde{b}}H^{ba} \, , \ 
\{H,H\}=0 \, , \ \widetilde{H}=1 \, ,
\end{equation}
 we have:
\begin{equation}
 \{f,g\}_H = 
-\{f,\{H,g\}\} \, .
 \end{equation}

The algebraic framework of derived brackets was introduced by 
Kosmann-Schwarzbach~\cite{yvette:derived}.
They were also independently considered by 
Voronov.

 \subsection{Language of $Q$-manifolds}
\label{subsec:Qmfd}
\subsubsection{$Q$-manifolds}
\label{subsubsec:Qmanifolds}

The theory of $Q$-manifolds was initiated by A. Schwarz~\cite{schwarz:semiclassical} (to whom belongs the name), M. Kontsevich~\cite{kontsevich:lect, Kontsevich_deform_quant}, and A. Vaintrob~\cite{vaintrob:darboux, vaintrob:normforms}.   See also~\cite{schwarz:aksz}.

\begin{definition}
A \emph{$Q$-manifold} $(M,Q)$ is a supermanifold equipped with a \emph{homological vector field}, which is an odd vector field $Q$ such that $Q^2=0$.  
\end{definition}

\begin{definition}
A \emph{$Q$-morphism} $(M_1,Q_1)$ to $(M_2,Q_2)$ is a 
 map of supermanifolds $\varphi : M_1 \rightarrow M_2$ such that 
\begin{equation} \label{eq:Q_phi_related}
    Q_1 \circ \varphi^* = \varphi^* \circ Q_2 \, ,
\end{equation}
where $\varphi^*$ is the pullback of functions.
\end{definition}

\subsubsection{$Q$-morphisms and $L_\infty$-morphisms of $L_\infty$-algebras}
\label{subsubsec:Qmanifolds_Linfty}

Let $V$ be a super vector space, $S^k V$ and $\Lambda^k V$ be the $k$th symmetric and exterior powers. An $L_\infty$-algebra is a concept motivated by physics, first mathematically formulated by Lada and Stasheff~\cite{lada_stasheff_L_infinity_algebras}
(in a $\ZZ$-graded antisymmetric version).

There are two versions of $L_\infty$-algebras;
one can be easily transformed into the other by a change of parity. We consider both here, as both  naturally occur in applications.

\begin{definition} 
\label{def:Linfty_antisym}
 The structure of an \emph{$L_\infty$-algebra in antisymmetric version} on a super vector space $L$ is defined by a collection of antisymmetric multilinear maps, the \emph{higher brackets}:
\[
l_k: \Lambda^k L \to L \, ,
\]
for $k \geq 0$, such that 
\begin{enumerate}[1)]
 \item the parity of the $k$th bracket is $k$ modulo $2$;
 \item   a series of \emph{higher Jacobi identities} is satisfied: for all $n$,
 \begin{equation}
     \sum_{k+l=n} \sum_{\sigma \in \text{Sh}(k,l)} \sign \sigma  \varepsilon(\sigma) (-1)^{k l}
     l_n \left( l_k \left( a_{\sigma(1)}, \dots, a_{\sigma(k)} \right) , a_{\sigma(k+1)}, \dots , a_{\sigma(k+l)} \right) = 0 \, ,
 \end{equation}  
  where $\varepsilon(\sigma)$ is the Koszul sign of a permutation $\sigma$, and  $\text{Sh}(k, l)$ is the set of \emph{shuffles}.
\end{enumerate}
\end{definition}
Recall that an $(k,l)$-shuffle is a permutation $\sigma$ of the set $\{1, 2, \dots, n\}$ such that:
\[
\sigma(1) < \sigma(2) < \dots < \sigma(k) \quad \text{and} \quad \sigma(k+1) < \sigma(i+2) < \dots < \sigma(k+l) \, .
\]
This means that the relative order of the first $k$ elements and the last $l$ elements is preserved, but they can interlace in any way.

\begin{definition}
\label{def:Linfty_sym}
 The structure of an \emph{$L_\infty$-algebra in symmetric version}  on a super vector space $V$ is defined by a collection of symmetric multilinear maps, the \emph{higher brackets}:
\[
l_k: S^k V \to V \, ,
\]
for $k \geq 0$, such that 
\begin{enumerate}[1)]
 \item all brackets are odd;
 \item a series of \emph{higher Jacobi identities} is satisfied: for all $n$,
 \begin{equation}
     \sum_{k+l=n} \sum_{\sigma \in \text{Sh}(k,l)} \varepsilon(\sigma)
     l_n \left( l_k \left( a_{\sigma(1)}, \dots, a_{\sigma(k)} \right) , a_{\sigma(k+1)}, \dots , a_{\sigma(k+l)} \right) = 0 \, .
 \end{equation}
\end{enumerate}
\end{definition}

(``Higher Jacobi identities'' in both versions  are a homotopy  generalization of the ordinary Jacobi identity for Lie algebras: first, the ordinary Jacobi for the binary bracket is modified by a term which is   chain homotopic to zero, and then further identities are seen as ``higher homotopies''.)

Now, the two versions are related by the formulas: $V=\Pi L$,
and 
\begin{equation}
 l_k^{\text{Sym}}(\Pi x_1, \dots , \Pi x_k) = (-1)^{\varepsilon} \Pi
 l_k^{\text{antiSym}}(x_1, \dots, x_k) \, .
\end{equation}
Here $\varepsilon=\sum \widetilde{x_i} (k-i)$.

The symmetric version of an $L_\infty$-algebra 
can be described by a single homological vector field 
\begin{equation}
 Q = Q^k(\xi) \der{}{\xi^k} \in \Vect(V) \, ,
\end{equation}
where $\xi^k$ are left linear coordinates on $V$ relative to some basis $\e_1, \dots, \e_n$.
The components of $Q$ can be found by evaluating brackets on the repeated even argument $\xi \in V$,
\begin{equation}
  Q(\xi) = \sum_{k=0}^\infty \frac{1}{k!} l_k(\xi,\dots, \xi) = \sum_{k=0}^\infty Q^k(\xi) \e_k \, .
\end{equation}

Conversely, given 
\begin{equation} \label{eq:Q_Taylor}
  Q = Q^k(\xi) \der{}{\xi^k} = \left( Q_0^k + \xi^i Q^k_i + \frac{1}{2} \xi^i \xi^j Q^k_{ji} + \dots  \right) \der{}{\xi^k} 
\end{equation}
we can define 
\begin{equation} \label{Linfty_sym}
  l_k(u_1, \dots, u_k) = [\dots[Q,u_1], \dots, u_k](0)
\end{equation}
for all $k \geq 0$ and $u_i \in V$. We identify vectors from $V$ with constant coefficients vector fields. For $L$,
\begin{equation} \label{Linfty_antisym}
 \iota \left( l_k(x_1, \dots, x_k) \right) =
(-1)^{\varepsilon}
 [\dots[Q,\iota(x_1)], \dots, \iota(x_k)](0)
\end{equation}
for all $k \geq 0$ and $x^i e_i \in L$. Here
$\iota(x)=(-1)^{\widetilde{x}} x^i \der{}{\xi^i}$.

\begin{theorem}
 Formulas~\eqref{Linfty_sym} and~\eqref{Linfty_antisym} define  symmetric and antisymmetric versions of $L_\infty$-algebra structure
 on $V$ and $L$ if and only if $Q^2=\frac{1}{2}[Q,Q]=0$.
\end{theorem}

(Formulas~\eqref{Linfty_sym} and~\eqref{Linfty_antisym} are examples of higher derived brackets~\cite{tv:higherder}. As a reference, one can see~\cite{Durham_Graded_Q_THV19}.)

\begin{remark} Up to a sign, the coefficients $Q^k_0, Q^k_i, Q^k_{ji}, \dots$ in~\eqref{eq:Q_Taylor} are the structure constants
 of the $0$-ry, unary, binary, and so on, brackets. 
\end{remark}

An \emph{$L_\infty$-morphism} $F$
of two $L_\infty$-algebras, $(V_1, \{l_k\})$ and $(V_2, \{l'_k\})$, in the  symmetric 
version for concreteness, consists of a collection of maps
$\varphi_k: S^k V_1 \to V_2$, for $k \geq 1$, satisfying a sequence of 
relations. We give these relations for even elements $u_i \in V_1$ (then they can be 
extended by linearity): 
\begin{multline*}
 \sum_{k=0}^{n}
\sum_{\sigma \in (n-k,k)-\text{shuffles}}
\varphi_{k+1}
\left( 
l_{n-k}
\left( u_{\sigma(1)}, \dots , u_{\sigma(n-k)} \right), u_{\sigma(n-k+1)}, \dots , 
u_{\sigma(n)} 
\right) \\
=
\sum_{r=1}^n
\sum_{\substack{i_1 + \dots + i_r = n \\
i_1 > 0, \dots, i_r > 0}}
\sum_{\substack{\text{combinations} \ \tau \\ \text{of} \ i_1, \dots, i_r \\
\text{each symmetric term only once}}} 
l'_{r} \left(\varphi_{i_1}(u_{\tau(1)}, \dots, 
u_{\tau(i_1)}), \dots , \varphi_{i_r}(u_{\tau(i_{r-1}+1)}, \dots, 
u_{\tau(i_r)})\right) \, .
\end{multline*}
(Here for simplicity we have assumed that in both algebras the $0$-ary brackets are 
zero.)

The combinations of $i_1, \dots, i_r$ are such that, when we have several $ 
l'_{r} \left(\varphi_{i_1}(u_{\tau(1)}, \dots, u_{\tau(i_1)}), \dots \right)$, whose 
values are the same because of the symmetry of $l$'s and $\varphi$'s, we count them 
only once.

The first few relations are listed below.
\begin{align*}
 n=1:& \quad
 \varphi_1(l_1(u)) =l'_1(\varphi_1(u)) \, ,\\
 n=2:& \quad
 \varphi_1(l_2(u_1, 
u_2)) + \varphi_2(l_1(u_1), u_2) + \varphi_2(l_1(u_2),u_1) =l'_1(\varphi_2(u_1, 
u_2)) + l'_2(\varphi_1(u_1), \varphi_1(u_2))\, , \\
 n=3:& \quad 
 \varphi_1(l_3(u_1, u_2, u_3)) \\
 & + \varphi_2(l_2(u_1, u_2), u_3)+ \varphi_2(l_2(u_1, 
u_3), u_2) + \varphi_2(l_2(u_2, u_3), u_1) \\
& 
+ \varphi_3(l_1(u_1), u_2, u_3) + \varphi_3(l_1(u_2),u_1, u_3) + 
\varphi_3(l_1(u_3), u_1, u_2) \\
& = l'_1(\varphi_3(u_1, u_2, u_3)) \\
&+l'_2(
\varphi_1(u_1),\varphi_2(u_2, u_3))+ 
l'_2(\varphi_1(u_2),\varphi_2(u_1, u_3))
+l'_2(\varphi_1(u_3),\varphi_2(u_1, u_2)) 
\\
&+ l'_3(\varphi_1(u_1), \varphi_1(u_2), \varphi_1(u_3))\, .
\end{align*}

The language of $Q$-manifolds  makes it possible to express these equalities concisely 
in a single line. It is convenient to work with them this way, see e.g. the famous 
work of Kontsevich~\cite{Kontsevich_deform_quant}. Namely,
consider super vector spaces $V_1$   and $V_2$ as supermanifolds,   and let $Q_1$ and 
$Q_2$ correspond to the respective $L_\infty$-structures (symmetric version). Then  an 
$L_\infty$-morphism between these $L_\infty$-algebras is exactly the same as  a 
$Q$-morphism between the $Q$-manifolds $(V_1, Q_1)$   and $(V_2, Q_2)$: more 
precisely, the  multilinear maps comprising an $L_\infty$-morphism are obtained as the 
terms of the Taylor expansion of a supermanifold map $V_1\to V_2$. See more about this 
in~\cite{Durham_Graded_Q_THV19}.

\subsubsection{Homotopy Poisson brackets: $P_\infty$ and $S_\infty$ structures}
\label{subsubsec:Sinfty}

\begin{definition}
 An \emph{$P_\infty$-algebra} is an antisymmetric 
$L_\infty$-algebra (Def.~\ref{def:Linfty_antisym}) also equipped with an 
even associative  multiplication  which is commutative in the super sense, 
and the brackets satisfy the Leibniz rule:
\begin{equation}
    \{a_1, \dots ,a_k, bc\} = \{a_1, \dots ,a_k, b\}c + (-1)^{(\widetilde{a_1}+\dots 
+ 
\widetilde{a_k}+k+1) \widetilde{b}} b \{a_1, \dots ,a_k, c\}
    \, . 
\end{equation} 
\end{definition}

\begin{definition}
 An \emph{$S_\infty$-algebra} is a symmetric 
$L_\infty$-algebra (Def.~\ref{def:Linfty_sym}), 
also equipped with an 
even associative  multiplication  which is commutative in the super sense, 
and the brackets satisfy the Leibniz rule:
\begin{equation}
    \{a_1, \dots ,a_k, bc\} = \{a_1, \dots ,a_k, b\}c + (-1)^{(\widetilde{a_1}+\dots + 
\widetilde{a_k}+1) \widetilde{b}} b \{a_1, \dots ,a_k, c\}
    \, . 
\end{equation} 

\end{definition}
Note that $P_\infty$ and 
 $S_\infty$ structures cannot be transformed one into another. 

 Note that these higher structures are not merely abstract generalizations; they 
naturally arise in ordinary Poisson geometry. Cattaneo and 
Felder~\cite{cattaneo_felder_p_infty} discovered $\Pinf$-structures through 
the study of deformations of coisotropic submanifolds in an ordinary Poisson 
manifold.

At the end of Subsec.~\ref{subsec:Poissonbrackets},
we mentioned that even and odd Poisson
brackets on $M$ can be obtained as derived brackets. 
If we replace Poisson bi-vector $P$
as in~\eqref{eq:bivector}
with an even multivector, which we write as a formal expansion 
\begin{equation}
 P = P_0 + P^a x_a^* + \frac{1}{2!} P^{ab} x_b^* x_a^* + \frac{1}{3!} P^{abc} x_c^* 
x_b^* x_a^* + \dots \, , \ \lsch P,P\rsch=0 \, , \
\widetilde{P}=0 \, ,
\end{equation}
we can generate on $M$ a series of antisymmetric brackets of alternating parities by 
using Voronov's ``higher derived brackets'' construction~\cite{tv:higherder}: 
 \begin{equation}
  \{f_1, \dots, f_k \}_P := \lsch \dots \lsch P,f_1\rsch, \dots ,f_k\rsch\big|_M \, , 
\quad 
\big|_M=\big|_{x^*=0} \, .
 \end{equation}
They will satisfy Leibniz rule by construction, and the condition $\lsch P,P\rsch=0$ 
will encode the higher Jacobi identities of $L_{\infty}$-algebra.
As Schouten brackets are derivations of degree $-1$, then only 
the 
term of degree $i$ contributes to the $i$-ary bracket. In particular, 
\begin{equation}
 \{ x^a \} = -P^a \, , \ \{ x^a, x^b \} 
= -(-1)^{\widetilde{a}}P^{ab}
 \, , \ \{ x^a, x^b, x^c \} = \pm P^{abc} 
 \, , 
 \dots
 \end{equation}
Thus, there is a 1-1-correspondence between $P_{\infty}$-structures on $\fun(M)$ and 
such multivector fields $P$. 
 
Similarly, we can consider an odd
Hamiltonian 
\begin{equation}
 H=H_0 + H^a p_a + \frac{1}{2!} H^{ab} p_b p_a + \frac{1}{3!} 
H^{abc} p_c p_b p_a + \dots \, , \ \{H,H\}=0 \, , \ \widetilde{H}=1 \, ,
\end{equation}
and generate on $M$ a series of odd symmetric brackets giving an 
$S_\infty$-structure on $\fun(M)$: 
 \begin{equation}
  \{f_1, \dots, f_k \}_H := \{\dots \{H,f_1\}, \dots ,f_k\}\big|_M \, , \quad 
\big|_M=\big|_{p_a=0} \, .
 \end{equation}
Only $H_i$ contributes to the $i$-ary bracket. In particular, 
\begin{equation}
 \{ x^a \} = \pm H^a \, , \ \{ x^a, 
x^b \} = \pm H^{ab}
 \, , \ \{ x^a, x^b, x^c \} = \pm H^{abc}  \, , 
 \dots
 \end{equation}
Again, there is a 1-1-correspondence between  $S_\infty$-structures on $\fun(M)$
and such odd Hamiltonians.

We can refer to a $P_\infty$-structure as an ``homotopy Poisson structure'' and to an 
$S_\infty$-structure as a ``homotopy odd Poisson structure'' or ``homotopy Schouten 
structure''.

\subsubsection{$Q$-morphisms and morphisms of Lie algebroids}
\label{subsec:Q_Lie_alg}

\begin{definition} \label{def:Liealgebroid}
  A vector bundle $E$ over a (super)manifold $M$ is a \textit{Lie algebroid }
if there is a Lie algebra structure on its sections $\Gamma(M,E)$
along with a bundle map, called an \emph{anchor}, $a: E \to TM$, such that 
for all $u,v \in \Gamma(M,E)$
and $f \in C^\infty(M)$
\begin{align*}
    [u,fv] &= a(u)(f)v+(-1)^{\widetilde{u}\widetilde{f}}f[u,v] \, , \\
    a([u,v]) &= [a(u),a(v)] \, .
\end{align*}
\end{definition}

A Lie algebroid $E$ also 
manifests itself on the neighboring spaces $E^*, \Pi E, \Pi E^*$ (see about this viewpoint in~\cite{ThThQ2012}).
We will be interested in the manifestation on  
$E^*$ and $\Pi E^*$.

Let $x^a$ be coordinates on $M$, and $u_i$, $\xi^i$, $\eta_i$ are the coordinates in the fibers corresponding to a choice of a basis $e_i$ in $E$. On $\Pi E$, a Lie algebroid can be described by a homological vector field of weight $+1$ (as introduced by Vaintrob~\cite{Vaintrob97}) of the form
\begin{equation}
\label{eq:Lie_alg_Q}
    Q=\xi^i Q_i^a(x)\frac{\partial}{\partial x^a} + \frac{1}{2} \xi^i \xi^j Q^k_{ji}(x)  \frac{\partial}{\partial \xi^k} \, ,
\end{equation}    
Then, on $E$, we have:
\begin{equation}
a(e_i)=Q_i^a \frac{\partial}{\partial x^a} \, , \quad     
[ e_i , e_j ]  = (-1)^{\widetilde{j}}Q_{ij}^k(x)e_k \, ,
\end{equation}
and on $E^*$, we have an even Poisson bracket of weight $-1$ given by:
\begin{equation}
\{ x^a , x^b \}  = 0 \, , \quad 
\{ u_i , x^a \}  = Q^a_i(x)
\, , \quad 
\{ u_i, u_j \}  = (-1)^{\widetilde{j}} Q_{ij}^k(x) u_k
\, .
\end{equation}
Finally,  on $\Pi E^*$, we have an odd Poisson bracket of weight $+1$ with 
\begin{equation}
\{ x^a , x^b \}  = 0
\, , \quad 
\{ \eta_i , x^a \}  = Q^a_i(x)
\, , \quad 
\{ \eta_i, \eta_j \}  =(-1)^{\widetilde{j}}  Q_{ij}^k(x) \eta_k
\, .
\end{equation}

\begin{example} The tangent Lie algebroid (on $TM$) 
is formed by the Lie bracket of vector fields, with the identity for the anchor. On $\Pi TM$, it manifests as the exterior differential $d$ and on $\Pi T^*M$, as the Schouten bracket 
of multivector fields. On $T^*M$, the tangent algebroid appears 
through the canonical Poisson bracket of Hamiltonians. 
\end{example}

\begin{example}
Given a Poisson structure on $M$ defined by $P=\frac{1}{2} P^{ij}x_j^* x_i^*$, the cotangent Lie algebroid
 on $ T^*M$ is described by the Koszul bracket of differential forms. The bracket of the sections is given by the Koszul brackets of $1$-forms, and the anchor is the Koszul bracket between a $1$-form and a function.
On $\Pi T^*M$, the cotangent algebroid manifests through the 
odd Hamiltonian vector field: 
\begin{equation}
    Q = -x_i^* P^{ia} \frac{\partial}{\partial x^a} + \frac{1}{2}
    x_i^* x_j^*  \frac{\partial P^{ij}}{\partial 
x^k}\frac{\partial}{\partial \xi^k} = d_P = \lsch 
P,-\rsch\, .
\end{equation}
On $\Pi TM$, it manifests as the Koszul bracket.
\end{example}

\begin{definition}
 A \emph{Lie algebroid morphism} between Lie algebroids $E_1$ and $E_2$ over the same base $M$ is 
 a vector bundle map that preserves the bracket and the anchor.
 \end{definition}
 In particular, as a vector bundle map, it is a fiberwise linear map over $M$.

 \begin{theorem}~\cite{Vaintrob97}
  A fiberwise linear map $E_1 \rightarrow E_2$ over $M$ is a Lie algebroid morphism if and only if the induced fiberwise linear map $\Pi E_1 \rightarrow \Pi E_2$ is a $Q$-morphism.
 \end{theorem}

\subsubsection{$Q$-morphisms and morphisms of $L_\infty$-algebroids}
\label{subsec:Q_Lie_infty_alg}
 
\begin{definition}[~\cite{tv:higherpoisson}]
An \textit{$L_\infty$-algebroid}
is a vector bundle $E$ over a (super)manifold $M$ with an $L_\infty$-structure on its sections along with a sequence of multilinear maps into the tangent bundle, known as  \emph{higher anchors} $a_n$, such that 
for every $n$, we have
\begin{equation}
    [u_1,\dots,u_{n-1},fv] = a_n(u_1,\dots,u_{n-1})(f)v+(-1)^{\left({\widetilde{u_1} +\dots +\widetilde{u_{n-1}}}\right)\widetilde{f}}f[u_1,\dots,u_{n-1},v] 
\end{equation}
for all $u_1,\dots,u_{n-1},v \in \Gamma(M,E)$
and $f \in C^\infty(M)$. 
\end{definition} 

One can define morphisms between $L_\infty$-algebroids (over the same base) similarly to the cases of Lie algebroids and $L_\infty$-algebras, i.e as a sequence of multilinear vector bundle maps satisfying certain relations with respect to the brackets of sections and anchors, but it is easier to work with the following shortcut compactly containing all these relations~\cite{tv:microformal}\,:

\begin{definition}
 An \emph{$L_\infty$-morphism} between $L_\infty$-algebroids $E_1$ and $E_2$ over the same base $M$ is a fiberwise
 (in general non-linear) map $\Pi E_1 \rightarrow \Pi E_2$ which is a $Q$-morphism.
\end{definition}

\subsection{Thick morphisms and morphisms of $L_\infty$-algebroids}
\label{subsec:thickmorphism} 
 
 A \emph{thick morphism}~\cite{tv:microformal} $\Phi: M_1 \tto M_2$ is a
 (formal) Lagrangian submanifold in $-T^*M_2 \times T^* M_1$
 defined by a \emph{generating   function} $S$, which is
a formal series in the target momentum coordinates $q_i$, with the coefficients being smooth functions of the source position coordinates $x$,
\begin{equation} \label{eq:S}
    S(x,q) = S^0(x) + \varphi^i(x) q_i
    + \frac{1}{2} S^{ij}(x) q_j q_i + \frac{1}{3!} S^{ijk}(x) q_k q_j q_i + \dots \ , 
\end{equation}
where indices are assumed to be symmetric.
Here the local coordinates in $M_1$ and $M_2$ are $x^a$ and $y^i$; and the corresponding coordinates in the fibers of the cotangent bundles $T^*M_1$ and $T^*M_2$ are $p_a$ and $q_i$, respectively.

The  \emph{pullback by a thick morphism} $\Phi$ is  a formal map $\Phi^* :  C^\infty (M_2) \rightarrow C^\infty (M_1)$ between the corresponding spaces of smooth functions regarded as infinite-dimensional supermanifolds,
\begin{align}
\label{eq:thick_morphsm_original_formula}
\Phi^* [g](x) &= g(y) + S(x,q) - y^i q_i \, , \\
\label{eq:thick_morph_diffeq1}
q_i &= \frac{\partial g}{\partial y^i} (y) \, , \\
\label{eq:thick_morph_diffeq2}
y^i &= (-1)^{\widetilde{i}} \frac{\partial S}{\partial q_i} (x,q) \, .
\end{align}
Note that we treat a super vector space as a supermanifold, as here, only even vectors are considered, but they may depend on auxiliry odd parameeters.

Equality~\eqref{eq:Q_phi_related} implies that every $Q$-morphism $\varphi: M_1 \rightarrow M_2$ induces a chain map $\varphi^*$ between the corresponding complexes $\bigl(\fun(M_2),Q_2\bigr) \to \bigl(\fun(M_1),Q_1\bigr)$. By a ``complex'' we mean just a super vector space equipped with an odd differential.  A chain map is linear by the definition. What is happening in the case of a thick morphism? One needs to introduce a  nonlinear chain map.

\begin{definition}~\cite{tv:tangent_functor}
Let $(V,d_1)$ and $(W,d_2)$ be two chain complexes, i.e. (super) vector spaces $V$ and $W$, respectively equipped with odd operators $d_1$ and $d_2$, both of which have square zero. 
A \textit{non-linear chain map} $f: V \rightarrow W$ is a $Q$-morphism of $Q$-manifolds 
$(V,d_1)$ and $(W,d_2)$.
\end{definition}

The polarization process, through a Taylor expansion, assigns to non-linear map $f: V \rightarrow W$, a series of even symmetric maps 
\begin{equation} \label{eq:fk}
 f_k : V^k \rightarrow W \ . 
\end{equation}
The $k$-th coefficient of the Taylor expansion of $f$ is the restriction of $f_k$ to the diagonal. 

\begin{theorem}~\cite{tv:tangent_functor}
The maps~\eqref{eq:fk} obtained from a non-linear chain map satisfy the following relations:
\begin{align*}
   &d(f_0)=0 \, , \\
   &d(f_1(v))=f_1(d(v)) \, , \\
   &d(f_2(v_1,v_2))=f_2(d(v_1),v_2)+(-1)^{\widetilde{v_1}} f_2(v_1,d(v_2)) \, , \\
   &d(f_3(v_1,v_2,v_3))=f_2(d(v_1),v_2,v_3)+(-1)^{\widetilde{v_1}} f_3(v_1,d(v_2),v_3)+(-1)^{\widetilde{v_1}+\widetilde{v_2}} f_3(v_1,v_2, d(v_3)) \, , \\
   & \dots 
\end{align*}
which means that Leibniz rule is satisfied for each $k \geq 2$.
\end{theorem}

Similarly, a non-linear chain map $f: \Pi V \rightarrow \Pi W$ is equivalent to a series of anti-symmetric multi-linear maps $V^k \rightarrow W$ of alternating parities.

Now, the fact that a $Q$-morphism $\varphi$ induces a chain map $\phi^*$ generalizes as follows.

\begin{definition}~\cite{tv:tangent_functor}
 An even thick morphism $\Phi: M_1 \tto M_2$ is a thick $Q$-morphism if
 \begin{equation}
Q^a_1(x) \der{S}{x^a}(x,q) = Q^b_2 \left(y^i=(-1)^{\widetilde{q_i}} \der{S}{q_i} \right) q_b   \, ,
 \end{equation}
 where $S=S(x,q)$ is an even generating function.
 An odd thick morphism $\Phi: \Pi M_1 \oto \Pi M_2$ is a thick $Q$-morphism if
 \begin{equation}
Q^a_1(x) \der{S}{x^a}(x,q) = Q^b_2 \left(\der{S}{y_i^*} \right) y_i^* \, ,
 \end{equation}
 where $S=S(x,y^*)$ is an odd generating function.
\end{definition}

\begin{theorem}~\cite{tv:tangent_functor}
 Every even thick $Q$-morphism $\Phi: M_1 \tto M_2$ induces a non-linear chain map
 \begin{equation}
  \Phi^*: \fun(M_2) \rightarrow \fun(M_1) \, .
 \end{equation}
 Every odd thick $Q$-morphism $\Phi: \Pi M_1 \oto \Pi M_2$ induces  a non-linear chain map
 \begin{equation}
  \Phi^*: \fun(\Pi M_2) \rightarrow \fun(\Pi M_1) \, .
 \end{equation}
\end{theorem}

\begin{definition}~\cite{TVnonlinearpullback}
Given a thick morphism $\Phi: M_1 \tto M_2$, Hamiltonians $H_1 \in C^\infty(T^*M_1)$ and $H_2 \in C^\infty(T^*M_2)$ are called $\Phi$-related if and only if
\begin{equation}
H_1\left(x, \der{S}{x}(x,q)\right)
= 
H_2\left((-1)^{\widetilde{q_i}}\der{S}{q_i}(x,q), q\right)
\end{equation}

\end{definition}

\begin{theorem}
\cite{TVnonlinearpullback}
\label{thm:main_thick}
Let Hamiltonians $H_1 \in  C^\infty(T^* M_1)$ and $H_2 \in  C^\infty(T^* M_2)$ define 
$S_\infty$-structures on
$C^\infty(M_1)$ and on $C^\infty(M_2)$, respectively.

If $H_1$ and $H_2$ are $\Phi$-related, then:
\begin{enumerate}
 \item  the Hamilton-Jacobi vector fields $X_{H_i}\in \Vect(\fun(M_i))$ are $\Phi^*$-related;
 \item  $\Phi^*: C^\infty(M_2) \rightarrow
C^\infty(M_1)$ is an $L_\infty$-morphism.
\end{enumerate}

\end{theorem}

\begin{remark}
    The opposite direction statement is also true. 
\end{remark}

\subsection{Algebraic theory of (formal) $\hbar$-differential operators}
\label{subsec.hdif}

The algebraic theory of (formal) $\hbar$-differential 
operators 
has been developed in~\cite{shemy:koszul}.
Such operators are used for generating higher
brackets structures. (Some close considerations can be found in~\cite{behrend-peddie-xu:2023}.)

Informally, a $\hbar$-differential operators are infinite-order differential 
operators with smooth function coefficients, which are also power series in $\hbar$. 
We also require that every partial derivative comes with a factor of 
$\hbar$ in its coefficient. For example, the coefficients of third-order 
derivatives must be divisible by at least 
$\hbar^3$.

We use notation 
\begin{equation*}
\hat p_a=-i\hbar \der{}{x^a} \, , 
\end{equation*}
which is
inspired by the  momentum operators  in quantum mechanics. Then a \emph{formal 
$\hbar$-differential operator} is a formal non-commutative formal power series
\begin{equation}\label{eq.formalL}
    L=L^0(\hbar,x)+L^a(\hbar,x)\hp_a  + 
L^{a_1a_2}(\hbar,x)\hp_{a_1}\hp_{a_2} +\ldots
\end{equation}
where the coefficients are smooth functions that are also formal power series in $\hbar$, considered together with the ``Heisenberg commutation relation''
\begin{equation}\label{eq.heis}
    [\hat p_a,f]=-i\hbar \der{f}{x^a}\,.
\end{equation}

Define the \emph{total degree} as the cumulative
degree in $\hbar$ and $\hp_a$'s.
The grading is invariant under the change of variables due to the Heisenberg 
relation~\eqref{eq.heis}. Since the commutation relations preserve the total degree 
(i.e., an $\hbar$ appears each time a partial derivative operator disappears), we 
obtain a grading on the 
$\hbar$-differential operators, as opposed to the usual filtration of differential 
operators.

It makes formal
$\hbar$-differential operators a graded algebra. 
Re-writing~\eqref{eq.formalL} as the sum of homogeneous 
components, we get
\begin{multline}\label{eq.formalL2}
    L=\sum_{n=0}^{+\infty} L^{[n]} 
    =
    \sum_{n=0}^{+\infty}\Bigl(L_0^{a_1\ldots 
a_n}(x)\hp_{a_1}\ldots\hp_{a_n}+
    (-i\hbar)L_1^{a_1\ldots 
a_{n-1}}(x)\hp_{a_1}\ldots\hp_{a_{n-1}}+
    \ldots +(-i\hbar)^nL_n^0(x)\Bigr)  \, ,
\end{multline}
where the coefficients are smooth functions of $x$ and have 
no dependency on $\hbar$. 

One can show~\cite{shemy:koszul} that the natural action of 
formal $\hbar$-differential operators on the algebra
of smooth functions which are also 
formal power series in $\hbar$, is well defined and 
respects  grading. There is a well-defined  action of 
formal $\hbar$-differential operators on  functions of the 
form $e^{\ih \la g(x)}$, where $g(x)$ is a formal power 
series in $\hbar$,   which  gives   products   of $e^{\ih 
\la g(x)}$ with   formal power series in both $\hbar$ and 
$\la$.

Since $\hbar$-differential operators are technically infinite-order 
differential operators, the question arises of how to define their symbol, as there 
is no highest order. The answer is to take the top-order term in each homogeneous 
component.
\begin{definition}
The \emph{principal symbol} of a formal 
$\hbar$-differential 
operator $L$ given by~\eqref{eq.formalL2} is the 
following formal function on $T^*M$:
\begin{equation}\label{eq.Lmodh}
    \symb(L) = L\pmod \hbar  = \sum_{n=0}^{+\infty} 
L_0^{a_1\ldots a_n}(x)\,p_{a_1}\ldots\,p_{a_n}\, .
\end{equation}
\end{definition}
That is the principal symbol of a formal 
$\hbar$-differential operator is a formal
Hamiltonian.
It was proved~\cite{shemy:koszul} that this function is 
well-defined.


\begin{lemma}~\cite{shemy:koszul}
\label{lem.symbcommut} For formal $\hbar$-differential 
operators,
\begin{equation}\label{eq.symbprod}
    \symb(AB)=\symb(A)\symb(B)\, .
\end{equation}
The commutator $[A,B]$ is always divisible by $\hbar$ and 
\begin{equation}\label{eq.symbcommut}
    \symb(i\hbar^{-1}[A,B])=\{\symb(A),\symb(B)\}\,,
\end{equation}
where at the right-hand side there is the Poisson  bracket on $T^*M$. 
\end{lemma}


The  following definitions were introduced by Th.~Voronov~\cite{tv:higherder}  and    are a modification of the construction due to Koszul~\cite{koszul:crochet85}.
\begin{definition}
\label{def:quantumbr}
For an operator $L$ on an algebra,
the   \emph{quantum $n$-bracket} and 
the  \emph{classical $n$-bracket} ($n=0,1,2,3, \ldots$) 
generated by $L$ are respectively
\begin{align}
\{f_1,\ldots,f_n\}_{L,\hbar}&:=  
 (-i\hbar)^{-n}  \left[\ldots [L,f_1],\ldots,f_n\right](1)
\label{eq.qubrack} \, , \\
   \{f_1,\ldots,f_n\}_{L} &:=
   (-i\hbar)^{-n}  \left[\ldots 
[L,f_1],\ldots,f_n\right](1)\pmod \hbar \label{eq.clbrack} 
\, .
\end{align}
\end{definition}

Here $f_i$ are functions on a supermanifold or 
elements of an abstract commutative superalgebra.
In order to avoid negative powers of 
$\hbar$, we assume that any $n$-fold commutator 
$\left[\ldots [L,f_1],\ldots,f_n\right]$ 
in the above 
formulas is divisible by $(-i\hbar)^{n}$.
In particular, this is true for formal 
$\hbar$-differential operators on supermanifolds as defined 
here. 
 
 \begin{example}[$0$-, $1$- and $2$-brackets]
 The    quantum  $0$-bracket  is simply
 \begin{align}\label{eq.brack0}
    \{\void\}_{L,\hbar}&=L(1)\,;
 \intertext{for the  quantum $1$-bracket   take $[L,f](1)=L(f1)-(-1)^{\Lt\ft}fL(1)=L(f)-L(1)f$, hence}
\label{eq.brack1}
    \{f\}_{L,\hbar}&=i\hbar^{-1}\bigl(L(f)-L(1)f\bigr)\,;
\intertext{similarly, for the $2$-bracket one has}
\label{eq.brack2}
   \{f,g\}_{L,\hbar}&=-\hbar^{-2} \bigl(L(fg)-L(f)g-(-1)^{\Lt\ft}fL(g)+ L(1)fg\bigr)\,.
 \end{align}
\end{example}

Quantum brackets are themselves differential operators (\textit{not}
$\hbar$-differential operators!) of finite order $\geq 1$ in each argument, with coefficients for the derivatives of order $\geq 2$ proportional to $\hbar$.
Furthermore, it is known~\cite{tv:higherder} that for any 
any $n$, the $n$-bracket generates the $(n+1)$-bracket 
as the obstruction 
to the Leibniz rule:
\begin{multline}\label{eq.hleibniz}
    \{f_1,\ldots,f_{n-1},fg\}_{L,\hbar}=\{f_1,\ldots,f_{n-1},f\}_{L,\hbar}\,g+
    (-1)^{\e}f\,\{f_1,\ldots,f_{n-1},g\}_{L,\hbar}\\
    + (-i\hbar)\{f_1,\ldots,f_{n-1},f,g\}_{L,\hbar}\,,
\end{multline}
where $(-1)^{\e}=(-1)^{(\tilde L +\ft_1+\ldots +\ft_{n-1})\ft}$.
We can call this ``$\hbar$-modified Leibniz identity'' for the collection of the brackets.
Modulo $\hbar$ the 
extra terms  disappear  and 
the resulting classical brackets each satisfy the Leibniz rule. (And they seize to be linked via   identities~\eqref{eq.hleibniz}.)
Hence, classical $n$-brackets are multiderivations
and can be generated by a Hamiltonian. In~\cite{shemy:koszul}, we 
proved~\label{thm.symbol} 
that this Hamiltonian $H$ for the classical brackets 
generated by $L$ is  the principal symbol of $L$, 
$H=\symb(L)$.


\section{Statement of the problem}
\label{sec:statement_of_the_problem}
\subsection{
The de Rham to Lichnerowicz diagram, 
where Koszul brackets are mapped into Schouten brackets}
\label{subsec:diagram_for_bivectorP}
If $M$ is equipped with an even Poisson structure,
then there is a single binary odd bracket induced 
on differential forms 
$\O(M)=C^\infty(\Pi T^*M)$, 
called the \emph{(classical) Koszul bracket}. 
Koszul constructed~\cite{koszul:crochet85} 
an odd second-order differential operator 
acting on $\O(M)$ that generates this bracket. 
There is a linear map $a^*: \O(M)\to \Mult(M)$ that maps the Koszul bracket on 
forms to the canonical Schouten bracket on multivectors, 
denoted by $\lsch -,-\rsch$.  
This map essentially ``raises indices'' with the help of the Poisson tensor $P^{ab}$, 
\begin{equation}
    dx^a=P^{ab}x_b^* \, .
\end{equation}
The same map also maps the de Rham differential $d$ to the Lichnerowicz differential $d_P$. So a commutative diagram arises:
\begin{equation}\label{eq.classdiagram}
    \begin{CD} \Mult^k  (M)@>{d_P}>> \Mult^{k+1}(M)\\
                @A{a^*}AA         @AA{a^*}A\\
                \O^k(M)@>{d}>> \O^{k+1}(M) \,,
    \end{CD}
\end{equation}
with vertical arrows also preserving the brackets.


From the viewpoint of Lie algebroids, 
the explanation is as follows.
The anchor of the the Lie algebroid $T^*M$ 
is a fiberwise linear map $E=T^*M \to TM$ and
is a Lie algebroid morphism.
As such it has ``manifestations'', among which there are the following: 
\begin{enumerate}[(1)]
 \item a fiberwise linear map
 \begin{equation}
 a\co \Pi E=\Pi T^*M \to \Pi TM \, , 
 \end{equation}
 which is a \textit{$Q$-morphism}, i.e. 
intertwines 
Lichnerowicz differential $Q=d_P$ 
with
exterior differential $Q=d$;
 \item the dual map (also fiberwise linear)
 \begin{equation}
 a^{\du}\co \Pi T^*M\to \Pi TM 
 \end{equation}
has the same domain and codomain as $a$.  
\end{enumerate}
As it follows from more general results 
for Lie algebroids, the pullback
\begin{equation}
 (a^{\du})^* \co \fun(\Pi TM) \rightarrow \fun(\Pi T^*M)
\end{equation}
is a homomorphism of 
the corresponding ``Lie-Schouten bracket ''
to the canonical Schouten bracket, so in our case from $\O(M)$ to $\Mult(M)$ 
mapping the Koszul bracket to the Schouten bracket. 

In our particular case, $a^{\du}=\pm a$ (because of the property $P^{ab}=P^{ba}(-1)^{(\at+1)(\bt+1)}$).   The sign plays no role, so a single map serves both purposes: intertwines the differentials ($d$ and $d_P$) and maps the bracket to the bracket (Koszul to Schouten). This is a peculiarity of the classical Poisson case, which \emph{no longer holds} in the homotopy case.

\subsection{$\Pinf$-case: higher Koszul brackets and the de Rham to Lichnerowicz 
diagram, where higher Koszul brackets are not mapped into Schouten bracket}
\label{subsec:diagram_for_arbitrary_P}
For $\Pinf$-manifolds, Khudaverdian--Voronov~\cite{tv:higherpoisson} introduced 
``higher Koszul brackets'' on the algebra of differential forms $\O(M)=\fun(\Pi TM)$. 
They make an $\Sinf$-structure, i.e. a sequence of odd symmetric brackets satisfying 
the higher Jacobi identities (so giving an $\Linf$-algebra structure in the symmetric 
version) and also the Leibniz identities with respect to the usual multiplication. 
Higher Koszul brackets were described  by an odd Hamiltonian $H\in \fun(T^*(\Pi TM))$ 
and also interpreted as the ``homotopy Lie-Schouten'' brackets on $\Pi E^*=\Pi TM$ for 
an $\Linf$-algebroid structure on $E=T^*M$ induced by a $\Pinf$-structure on $M$.

In~\cite{tv:higherpoisson}, they also  found a formula generalizing the above fiberwise map
$a\co \Pi T^*M\to \Pi TM$ to the case of a general $P \in 
C^\infty(\Pi T^* M)$:
\begin{equation} \label{eq:KHV_map}
dx^a=(-1)^{\at+1}\der{P}{x^*_a} \, .    
\end{equation}
If $P$ is a bi-vector, this simplifies into the classical 
raising of indices formula.
It was shown that the pullback $a^*$ of 
such a non-linear map $a$ still intertwines $d$
and $d_P=\lsch P,-\rsch$. 
This gives the analog of the commutative diagram
of the classical case (not preserving the 
$\ZZ$-grading),
\begin{equation}\label{eq.nonclassdiagram}
    \begin{CD} \Mult  (M)@>{d_P}>> \Mult(M)\\
                @A{a^*}AA         @AA{a^*}A\\
                \O(M)@>{d}>> \O(M) \,.
    \end{CD}
\end{equation}
However, the pullback $a^*\co \O(M)\to \Mult(M)$, 
as a single linear map, can no longer map the entire 
sequence of higher Koszul 
brackets to the Schouten bracket. 
This raises a problem: how to construct 
 $\Linf$-morphism from $\O(M)$ with the higher Koszul 
brackets to $\Mult(M)$ 
with the Schouten bracket, which, in particular, should be a
non-linear map from forms to multivectors.

\subsection{$\Pinf$-case: an $L_\infty$-morphism
mapping higher Koszul brackets into Schouten bracket} 
\label{subsec:Linfty_mapping_HK_into_Schouten}

This problem was the primary motivation behind Voronov's 
development of the concept and apparatus of thick 
morphisms~\cite{tv:nonlinearpullback}.
An $\Linf$-morphism in question was constructed
in~\cite{tv:microformal}, \cite{tv:highkosz}.
Since this will serve as the foundation for our own 
construction, and the explanations in the original 
papers are brief, we provide here additional details.

Firstly, it was observed that formula~\eqref{eq:KHV_map} and 
the corresponding map $a^*$, which intertwines $d_P$ and $d$ 
for a general 
$P \in C^\infty(\Pi T^* M)$, 
can be interpreted as 
a manifestation of the anchor of the 
cotangent $L_\infty$-algebroid. Indeed, the anchor, an 
$L_\infty$-morphism $T^*M 
\rightsquigarrow TM$ (a proof of the fact that the anchor is 
an 
$L_\infty$-morphism, and also details about anchor's 
``manifestations'', can be 
found in~\cite{tv:microformal}), has $a$ as its 
manifestation in the form of a 
fiberwise $Q$-map. Namely, we have a $Q$-morphism
\begin{equation} \label{eq:a}
 a\co \Pi T^*M\to \Pi TM
\end{equation}
the pullback by which
\begin{equation}
 a^*\co C^\infty (\Pi TM) \to C^\infty(T^*M) 
\end{equation}
intertwines the homological vector fields, the manifestations of the
cotangent and tangent algebroids, i.e.
\begin{equation}
      d \circ a^* = a^* \circ d_P        \, .
\end{equation}
Unlike the classical case of  Lie algebroids, the bundle map
$a$ is nonlinear on fibers and does not have a dual in the 
conventional sense. However, its dual does exist as a thick 
morphism, 
\begin{equation}
a^{\du}\co \Pi T^*M\tto \Pi TM \, . 
\end{equation}
To construct $a^{\du}$, we recall that thick morphisms 
generalize usual smooth maps of manifolds, and an
ordinary map $\varphi: M_1 \to M_2$, $y^i=\varphi^i(x)$, is a thick morphism with the
generating  function $S=\varphi(x)^iq_i$
(linear in target momenta). 

To give an explicit local formula, let coordinates in $\Pi 
T^* M$ be denoted by $x^a$ and $x^*_a$. The corresponding 
conjugate momenta are denoted by $p_a$ and $\pi^a$. In $\Pi 
T M$, let the fiber coordinates be $dx^a$, with the 
corresponding momenta denoted
$\pi_a$. Then $a$, defined by formula~\eqref{eq:KHV_map}, 
can be considered as a thick morphism with the generating
 function
\begin{equation} \label{eq:S}
 S = S(x^a,x^*_a;p_a,\pi_a)= x^a p_a
 +(-1)^{\at+1}\der{P}{x^*_a}(x,x^*) \pi_a \, .
\end{equation}
To construct $a^{\du}$ as a thick morphism, one essentially 
applies Mackenzie-Xu (MX) transformation to $S$ to obtain 
$S^*$. However, since $S$ is not a genuine function, we cannot
simply apply the MX formulas directly to~\eqref{eq:S} to 
obtain $S^*$. Instead, we apply MX to the equations 
describing the corresponding thick morphism
as a Lagrangian submanifold in $T^*(\Pi T^* M) \times 
-T^*(\Pi TM)$ and express the result via a new generating function $S^*$.

The MX, for the $T^*(\Pi T^*M)$ part, is the following 
anti-symplectomorphism:
\begin{align}
  \MX \co T^*(\Pi T^* M) &\to T^*(\Pi T M) \notag \\
  (x^a,x^*_a,p_a,\pi^a) & \mapsto \bigl(x^a:=x^a, dx^a:=(-1)^{\widetilde{a}+1}\pi^a, p_a:=-p_a,\pi_a:=x^*_a\bigr)
  \label{eq:MXfromX}
\end{align}

For the $-T^*(\Pi TM)$ part, we use the second copy of $M$ (and everything related).
Our notations for the coordinates on $\Pi TM$ are $y^a$ and
$dy^a$, with the corresponding conjugate momenta   $q_a$ and
$\pi_{a,2}$, and for $\Pi T^*M$, the coordinates are $y^a$ and $y^*_a$, with the conjugate momenta   $q_a$ and $\pi^{a,2}$.
The MX  then will be
\begin{align}
  \MX \co   T^*(\Pi T M) & \to T^*(\Pi T^* M) \notag \\
  (y^a,dy^a,q_a,\pi_{a,2}) & \mapsto \bigl(y^a:=y^a,y^*_a:=\pi_{a,2}, q_a:=-q_a, \pi^{a,2}:=(-1)^{\widetilde{a}+1}dy^a\bigr)
  \label{eq:MXfromH}
\end{align}

In this notation, the formula for the generating function~\eqref{eq:S}
becomes
\begin{equation}
 S = S(x^a,x^*_a;q_a,\pi_{a,2})= x^a q_a
 +(-1)^{\at+1}\der{P}{x^*_a}(x,x^*) \pi_{a,2}\,.
\end{equation}
The
corresponding Lagrangian submanifold in $T^*(\Pi T^* M) 
\times -T^*(\Pi TM)$ 
is defined by the following system of equations:
\begin{equation}
\begin{array}{c@{\hspace{2em}}c}
\begin{aligned}
    (-1)^{\widetilde{a}} \der{S}{q_a} &=y^a \\
    (-1)^{\widetilde{a}+1} \der{S}{\pi_{a,2}} &=dy^a 
\end{aligned}
&
\begin{aligned}
    \der{S}{x^a} &=p_a \\
    \der{S}{x_a^*} &=\pi^a
\end{aligned}
\end{array}
\end{equation}
or, after the substitution of the expression for $S$, 
\begin{equation}
\begin{array}{c@{\hspace{2em}}c}
\begin{aligned}
    x^a &=y^a \\
    (-1)^{\widetilde{a}} \der{P}{x_a^*}(x,x^*) &=dy^a 
\end{aligned}
&
\begin{aligned}
    q_a+(-1)^{\widetilde{a}} \frac{\partial^2 P }{\partial 
x^a \partial x_b^*}(x,x^*) \pi_{b,2} &=p_a \\
    (-1)^{\widetilde{a}} \frac{\partial^2 P }{\partial x_b^* 
\partial x_a^*}(x,x^*) \pi_{a,2} &=\pi^b \, .
\end{aligned}
\end{array}
\end{equation}
Applying MX to these equations, we obtain equations for a
Lagrangian submanifold in $T^*(\Pi T M) \times -T^*(\Pi 
T^*M)$,
\begin{equation}
\begin{array}{c@{\hspace{2em}}c}
\begin{aligned}
    y^a &=x^a \\
    (-1)^{\widetilde{a}} \der{P}{x_a^*}(x,\pi_a) 
&=(-1)^{\widetilde{a}+1} \pi^{a,2}
\end{aligned}
&
\begin{aligned}
    -q_a+(-1)^{\widetilde{a}} \frac{\partial^2 P }{\partial 
x^a \partial x_b^*}(x,\pi_c) y_b^* &=-p_a \\
    (-1)^{\widetilde{a}} \frac{\partial^2 P }{\partial x_b^* 
\partial x_a^*}(x,\pi_c) y_a^* &=(-1)^{\widetilde{b}+1}dx^b 
\, .
\end{aligned}
\end{array}
\end{equation}
We can derive a formula for $S^*=S^*(y,y^*,p_a,\pi_a)$, by 
guessing that these equations should be (not necessarily in 
the same order as given) the following general equations 
from the definition of a Lagrangian submanifold:
\begin{equation}
\begin{array}{c@{\hspace{2em}}c}
\begin{aligned}
    (-1)^{\widetilde{a}} \der{S^*}{p_a} &=x^a \\
    (-1)^{\widetilde{b}+1} \der{S^*}{\pi_{b}} &=dx^b 
\end{aligned}
&
\begin{aligned}
    \der{S^*}{y^a} &=q_a \\
    \der{S^*}{y_a^*} &=\pi^{a,2} \, .
\end{aligned}
\end{array}
\end{equation}
We obtain that
\begin{equation}
 S^*=S^*(y,y^*,p_a,\pi_a)=y^a p_a + (-1)^{\widetilde{a}}
 \der{P}{x_a^*}(x,\pi_b) y_a^* \, .
\end{equation}

Now, we define the \textit{dual} to $a$,
 \begin{equation}
a^{\du}\co \Pi T^*M \tto \Pi TM 
 \end{equation}
as the thick morphism with the generating   function
$S^*$. 

One can see that MX preserves the Hamilton-Jacobi
equation, and, therefore, since $a$ is an
$S_\infty$ thick morphism (as a $Q$-morphism), then
$a^\du$ is also an $S_\infty$ thick morphism~\cite{tv:microformal}. (In greater detail, the Hamilton-Jacobi equation in question expresses the condition for the   Hamiltonians corresponding to $d$ and $d_P$ to conicide on the Lagrangian submanifold in the product; hence, after applying MX, their images will concide on the image, i.e. the new Lagrangian submanifold corresponding to the dual thick morphism $a^{\du}$.) This implies, by Theorem~\ref{thm:main_thick},
that the pullback by this thick morphism,
\begin{equation}
(a^{\du})^*\co \fun(\Pi TM) \rightsquigarrow \fun(\Pi 
T^*M) 
 \end{equation}
gives the desired $\Linf$-morphism (from higher Koszul 
brackets to the Schouten bracket). 


\subsection{``BV operators'' generating higher brackets. The 
problem.  Higher Koszul to Schouten diagram}
\label{subsec:problem_generating_HK_and_what_we_want}

The goal of this paper is to lift the above concepts to the 
``quantum level''. 
By this we mean the following.

\subsubsection{Generating a single bracket by a second-order odd
differential 
operator.}
Odd brackets, like the canonical Schouten bracket (Ex.~\ref{ex:Schouten_bracket}) and 
the 
classical Koszul bracket (Ex.~\ref{ex:Koszul}), can be defined using odd 
differential operators of square zero, which we refer to 
here as \textit{BV operators}.

The corresponding algebraic notion is that of a \textit{Batalin-Vilkovisky 
(BV) algebra}, which is a commutative (super) algebra $A$ equipped
with an odd Poisson bracket and an odd differential operator $\Delta: A \rightarrow 
A$, $\ord \Delta \leq 2$, so that 
\begin{equation} \label{eq:BV}
 \D(ab)=\D(a)b +(-1)^{\tilde a}a\D(b) + [a,b] \, .
\end{equation}
The condition
$\Delta^2=0$ implies the Jacobi identity for the 
bracket. 
These structures were studied by
Lian-Zuckermann~\cite{lian-zuckerman:1993},
Getzler~\cite{getzler:bv1994} and 
Schwarz~\cite{penkava-schwarz:1994}.
See also 
Witten~\cite{witten:antibracket90}, who attracted attention to such an
algebraic structure in the work of Batalin and Vilkovisky on 
quantization of gauge theories~\cite{bv:perv,bv:vtor}. 
Odd 
differential operators that generate an odd bracket via a formula 
 like~\eqref{eq:BV} are therefore referred to as ``Batalin--Vilkovisky operators''.

\begin{example}
The canonical Schouten bracket (see Ex.~\ref{ex:Schouten_bracket})
on $\fun(T^*M)$ can be generated by the divergence 
operator $\Delta=\delta_{\rho}$. 
Recall that the divergence operator is:
\begin{equation*}
     \delta_\rho T=(-1)^{\tilde{a}}\frac{1}{\rho(x)}
      \frac{\partial}{\partial x^a}
      \left(\rho(x)\frac{\partial{T}}{\partial x^*_a}\right) \, ,
\end{equation*}
     defined using a choice of a volume element $\rho$ on $M$. Note, that although 
we are used to thinking about divergence as a first order operator, from the 
viewpoint of supermanifolds it is a differential operator of second order. 
 See e.g. Kirillov's survey~\cite{kirillov:invariant}
 and in the super language $T=T(x,x^*)$ see e.g.~\cite[Ch.5]{tv:git}.
\end{example}

\begin{example} The Koszul bracket  (see Ex.~\ref{ex:Koszul}) on 
$\fun(\Pi TM)$  for a Poisson manifold $M$ can be generated by   
an odd second-order operator
\begin{equation*}
\Delta=\partial_P=[d, i(P)] \, , 
\end{equation*}
where $i(P)$ is the operator of the interior product  with the bivector $P$:
\begin{equation*}
i(P)=\frac{1}{2} P^{ab} 
\frac{\partial}{\partial dx^b} \frac{\partial}{\partial 
dx^a} \, , 
\end{equation*}
see Koszul~\cite{koszul:crochet85}.
\end{example}

Khudaverdian~\cite{hov:deltabest} and 
Schwarz~\cite{schwarz:bv} independently proposed a formula for a BV operator 
generating a given odd bracket. 
The understanding that this formula can be applied in a more general setting is 
due to Kosmann-Schwarzbach and Monterde~\cite{yvette:divergence}, and in particular, 
it was shown in~\cite{yvette:divergence} that Koszul's operator 
$\partial_\rho$ can be obtained by that formula. 

%

%
%

%

A BV operator carries more 
information than the bracket it generates, and is not 
unique.
A construction of a BV operator is a quantization of the 
corresponding odd Hamiltonian, where the odd Hamiltonian is 
the principal symbol of the BV operator. 
%

\subsubsection{Generating higher brackets by a 
higher-order
differential 
operator: missing   Leibniz rule.}
For higher-order operators, Koszul's construction yields an 
$\Linf$-algebra structure, as described by 
Kravchenko~\cite{kravchenko:bv}, who also introduced the 
term $BV_{\infty}$ for this situation. See also~Bering~\cite{bering:higher}. However, in this 
case, the Leibniz identities do not hold. So, at the first glance, unlike the case of 
operators of second order and a binary odd bracket, it is not possible top obtain an 
$S_{\infty}$-structure from operators. 

Note   that the description of odd brackets using Hamiltonians and 
Koszul's construction   both fit   into the abstract 
framework of Voronov's higher derived 
brackets~\cite{tv:higherder,tv:higherderarb}.  


\subsubsection{Generating higher brackets by a
formal $\hbar$-differential operator.}
As we already mentioned (see~\ref{subsec.hdif} in the preliminaries section), in Voronov's~\cite{tv:higherder},  Planck's
constant $\hbar$ was   introduced  
into modified  Koszul's formulas 
to interpolate between the properties of the 
brackets generated by a higher-order operator and those of 
an $\Sinf$-structure (the brackets generated by a Hamiltonian), i.e. with and without derivation property..
This approach was fully developed in~\cite{shemy:koszul}. 
This theory 
makes it possible to speak about the principal symbol, in a 
new sense,  for 
operators of \emph{infinite order}.
In particular, there is a formal $\hbar$-differential 
operator $\Delta_P$ that generates higher Koszul 
brackets~\cite{tv:highkosz}.
This formula in a slightly different 
form appeared also in A.~Bruce's PhD thesis
and his subsequent work~\cite{bruce:tulczyjew2010}, 
where he proved that the operator generated the higher
Koszul brackets.
This proof there was limited to finite sums of homogeneous
multivectors.
In~\cite{shemy:koszul}, we were able to treat the general $P_{\infty}$-case due to the use of
formal
$\hbar$-pseudodifferential operators. 

Now that we have  $\hbar$-differential operators 
generating brackets (as ``lifts'' of the 
Hamiltonians), we would like to   
similarly  lift  $\Linf$-morphisms. 
Specifically, 
given the $\Linf$-morphism $(a^{\du})^*$,
described above,
which maps higher Koszul brackets to   Schouten bracket,
we want  \emph{to lift $(a^{\du})^*$ to an operator that 
would intertwine the
operator $\Delta_P$ with 
a BV operator for the Schouten bracket}.
Specifically, we 
will seek an operator for which the known $\Linf$-morphism 
is recovered as the 
classical limit as $\hbar\to 0$. 
(We first mentioned this problem 
in~\cite{shemy:koszul}, referring to the future operator as 
``dual quantum 
anchor''. ``Classical limit'' here is understood in the sense of quantum and classical thick morphisms, see below.)

One may view this   operator that we are looking for, as
a  ``quantum anchor''   for the cotangent 
$\Linf$-algebroid. (More precisely, the quantum anchor in one of the manifestations, ``on $\Pi E^*$'', if we think about an algebroid  $E$.)

Considering commutative diagrams
as discussed above, we can also describe this problem as 
passing from the
``de Rham to Lichnerowicz diagram''~\eqref{eq.nonclassdiagram} to a ``higher Koszul
to Schouten diagram''.

\section{Solution}
\label{sec:solution}
\subsection{The idea}

The basic idea is that since the operator
we are seeking is seen as one of the
manifestations of the ``quantum anchor'' in a
``quantum''  $L_\infty$-algebroid, it can be easier to first construct it in another manifestation and then
derive the one we need using the ``quantum
Mackenzie-Xu transformation''.

We will start with the $\Linf$-structure for 
$T^*M$. It has the following manifestations: 
\begin{enumerate}[1)]
 \item on $\Pi T^*M$ as the homological vector 
field $d_P$, and
\item on $\Pi TM$ as the $\Sinf$-structure given by the 
higher Koszul brackets,
\end{enumerate}
which are related by the Mackenzie-Xu transformation (up to 
a minus sign).
Correspondingly, there are two 
manifestations of the anchor, as an $L_\infty$-morphism $T^*M
\rightsquigarrow TM$:
\begin{enumerate}[1)]
 \item given by a $Q$-morphism $a\co \Pi T^*M\to \Pi TM$, 
and
 \item given by the thick morphism $a^{\du}\co \Pi
T^*M\tto \Pi TM$ that relates the corresponding 
Hamiltonians.
 \end{enumerate}
 
Our plan is to lift the entire
picture to the level of 
operators, 
where the Mackenzie-Xu 
transformation will be replaced by its quantum analog.
Constructing the quantum analog of 
$a$ will be more straightforward, and the quantum analog
of $a^{\du}$ will be obtained through the
application of the ``quantum Mackenzie-Xu'' transformation.

The resulting operators will take the 
form of a ``quantum pullback'' 
(i.e. the pullback by a 
``quantum thick morphism''~\cite{tv:oscil,tv:microformal}). 
for which the known $\Linf$-morphism is the classical 
limit as $\hbar\to 0$. 
(A quantum pullback is a formal $\hbar$-Fourier operator
of a particular type.) Below we  work out these steps.

\subsection{Quantum manifestations of the cotangent 
$\Linf$-algebroid.} Recall them
from~\cite{shemy:koszul}. From now on, we will use the 
approach referred to as 
 ``non-symmetric'' in~\cite{shemy:koszul}, which uses
operators acting on functions and at some point requires a choice of a volume element $\brh$. We
hope to consider the ``symmetric'' approach, with
operators on half-densities, elsewhere.

Given a $\Pinf$-structure on $M$, defined by 
$P=P(x,x^*) \in \fun(T^*M)$, it induces an $L_\infty$-algebroid
structure on the $T^*M$. 
We first recall how the Hamiltonian $H_P$, which generates the higher Koszul
brackets, is constructed, as this will be relevant. On $\Pi
T^*M$, the $L_\infty$-algebroid structure manifests through the Lichnerowicz differential
$d_P$. Consider its Hamiltonian lift, which is the Hamiltonian $H_{d_P} \in
\fun (T^*(\Pi T^*M))$ generating a degenerate case of an $\Sinf$-structure 
consisting solely of $d_P$  as a  unary bracket, $d_P(X)= \lsch P, X  \rsch \,$, for all $X \in \fun(\Pi T^*M)$ and where the bracket on the right hand
side is the Schouten bracket. Consider the
Hamiltonian $\HSch \in \fun (T^*(\Pi T^*M))$, which generates the Schouten 
bracket. Then $d_P(X)= \lsch \HdP, 
X  \rsch \, = \{ \{ \HSch, P \}, X \}$, where $\{-,-\}$ denotes the canonical Poisson
bracket. This implies $\HdP=\{ \HSch, P \}$. 
Denote by $\sta$ the action of 
 the classical MX-transformation. 
 Noting that 
MX-transformation is an anti-symplectomorphism, we obtain
\begin{equation} \label{eq:HP1}
H_P=-\{\HSch^\sta,P^\sta\} \, . 
\end{equation}

To find $\HSch^\sta$, recall the different manifestations of the Lie algebroid 
$TM$. On the 
$\Pi TM$, it manifests through the de Rham differential $d=dx^a \der{}{x^a}$, whose
Hamiltonian lift is $\Hd =dx^a p_a$. On $\Pi T^*M$, it is manifested through the 
Schouten
bracket. Thus, see~Fig.~\ref{dia:tangent},  the MX-transformation maps $\HSch$ into 
$\Hd$, leading to
\begin{equation}
\HSch^\sta=\Hd \, . 
\end{equation} 

\begin{figure}[ht]
 \begin{center}
\begin{tikzpicture}[thick, font=\small, node distance=1.2cm, 
                    every node/.style={align=center}]
    \node (a) {${\Pi TM:}$};  
    \node (b) [right=of a] {$d=dx^a \der{}{x^a}$};  
    \node (c) [right=of b] {};  
    \node (d) [right=of c] {$\Hd =dx^a p_a$};  %
    
    \node (2d)[below=0.7cm of d] {$\HSch=(-1)^{\at+1}\pi^a p_a$};
    \node (2b) [below=0.7cm of b] {\text{Schouten bracket}};
    \node (2a) [below=0.7cm of a] {${\Pi T^*M:}$};  
     \draw[->] (b) -- (d); 
     \draw[->] (2b) -- (2d); 
     \draw[->] (2d) -- (d) node[midway, right] {MX}; 
\end{tikzpicture}
\end{center}
\caption{The Lie algebroid $TM$ described by Hamiltonians.}
\label{dia:tangent}
\end{figure}

Similarly for $P^\sta$ for $P(x,x^*)$, the MX formulas~\eqref{eq:MXfromX}  give
simply    $P^{\sta}=P(x,\pi)$.
Thus,~\eqref{eq:HP1} implies
\begin{equation} \label{eq:H_P}
H_P= -\{ \Hd 
, P(x,\pi) \} \, .
\end{equation}
This process is illustrated in~Fig.~\ref{dia:cotangent}.
\begin{figure}[ht]
 \begin{center}
\begin{tikzpicture}[thick, font=\small, node distance=2cm, 
                    every node/.style={align=center}]
    \node (a) {${\Pi T^*M:}$};  
    \node (b) [right=of a] {$d_P$};  
    \node (c) [right=of b] {};  
    \node (d) [right=of c] {$\HdP = \{\HSch, P\}$};  %
    
    \node (2d)[below=0.7cm of d] {$H_P =-\{ \Hd 
, P(x,\pi) \}$}; 
    \node (2b) [below=0.7cm of b] {\text{higher Koszul brackets} ($S_\infty$)};
    \node (2a) [below=0.7cm of a] {${\Pi TM:}$};  
     \draw[->] (b) -- (d); 
     \draw[->] (d) -- (2d) node[midway, right] {MX}; 
     \draw[->] (2b) -- (2d);
\end{tikzpicture}
\end{center}
\caption{The $L_\infty$-algebroid $T^*M$ described by Hamiltonians.}
\label{dia:cotangent}
\end{figure}

Both of the structures in Fig.~\ref{dia:cotangent} can be lifted to (formal) 
$\hbar$-differential operators whose principal symbols are the Hamiltonians
encoding the classical structures and which have
square zero. This is what we refer to as lifting to  
the ``quantum level''.

On $\Pi T^*M$, we have $d_P$,
which is   a differential operator of 
first 
order. To make it an $\hbar$-differential 
operator, we take simply $-i\hbar\,
d_P$. This is a 
 generating operator of the $Q$-structure on $\Pi T^*M$ 
 regarded as a degenerate case of an $\Sinf$-structure 
 consisting only of $d_P$  as a  unary bracket:  
\begin{equation} \label{eq:D_P_hat}
   d_P(F)=\ (-i\hbar)^{-1}[
-i\hbar\, 
d_P,F](1)   \mod \hbar\,,
\end{equation}
for all $F\in \fun(\Pi T^*M)$. 
Such a generating operator is not unique
(a fact we make use of later). One can also introduce a 
free
term and consider
\begin{equation}
   D_{d_P}:=-i\hbar\, d_P   -i\hbar\, F_0\,,
\end{equation} 
where $F_0$ is some odd element in $\fun(\Pi T^*M)[[-i\hbar]]$.  Adding 
this term does not affect the principal symbol. The 
condition $(D_{d_P})^2=0$ is equivalent to 
\begin{equation}
d_P(F_0)=0 \, . 
\end{equation}
For instance, this condition is automatically 
satisfied for $F_0=d_P(F)$ for an even $F$.

On $\Pi TM$, we have an $S_\infty$-structure, namely, the
higher Koszul brackets. 
We quantize equality~\eqref{eq:H_P} into 
\begin{equation}
\label{eq:DeltaP}
\Delta_P= -\frac{i}{\hbar}  \Big[ 
-i\hbar d,\hat P \Big] = -[d, 
\hat{P}] \, ,
\end{equation}
where by $\hat P$ we denote a
formal $\hbar$-differential operator, whose principal symbol is the Hamiltonian lift 
of the multivector $P$. This Hamiltonian is $P(x^a, \pi_a)$, and so as $\hat P$ we can take $\hat
P=P(x,-i\hbar \frac{\partial}{\partial 
dx})$.
The higher Koszul brackets can be generated by $\Delta_P$ as the
quantum brackets taken modulo $\hbar$, see Def.~\ref{def:quantumbr}.
The operator 
$\Delta_P$    is not unique; but  it is distinguished 
because its construction as the  ``quantum Poisson
bracket'' exactly follows that of the 
classical Hamiltonian $H_P$. See Fig.~\ref{dia:hcotangent}.

\begin{figure}[ht]
 \begin{center}
\begin{tikzpicture}[thick, font=\small, node distance=1.5cm, 
                    every node/.style={align=center}]
    \node (a) {${\Pi T^*M:}$};  
    \node (b) [right=of a] {$d_P$};  
    \node (c) [right=of b] {};  
    \node (d) [right=of c] {$\HdP = \{\HSch, P\}$};  %
    \node (e) [right=of d] {$D_{d_P}=-i\hbar d_P-i\hbar F_0$}; 
    \node (2e)[below=0.7cm of e] {$\Delta_P=-[d, \hat{P}]$}; 
    \node (2d)[below=0.7cm of d] {$H_P =-\{ \Hd , P(x,\pi) \}$}; 
    \node (2b) [below=0.7cm of b] {$S_\infty$: \text{higher Koszul br}};
    \node (2a) [below=0.7cm of a] {${\Pi TM:}$};  
     \draw[->
, decorate, decoration={zigzag, segment length=2mm, amplitude=0.4mm}] (d) -- (e) 
node[midway, above] {$\hbar$};
     \draw[->, decorate, decoration={zigzag, segment length=2mm, amplitude=0.4mm}] 
(2d) -- (2e) node[midway, above] {$\hbar$}; 
     \draw[<->] (d) -- (2d) node[midway, right] {MX}; 
\end{tikzpicture}
\end{center}
\caption{Quantum manifestations of the $L_\infty$-algebroid $T^*M$.}
\label{dia:hcotangent}
\end{figure}

\begin{remark}
Note that in works~\cite{tv:highkosz}
and~\cite{shemy:koszul},
formula~\eqref{eq:DeltaP} appears with the opposite sign. It is one of those 
signs that depend on conventions. Our sign corresponds to the choice of the
MX transformation as an anti-symplectomorphism.

Let us verify our 
formula~\eqref{eq:DeltaP} in the case of 
$P=\frac{1}{2}P^{ab}x_b^*x_a^*$ (note that in Koszul's original notation, the indices of 
$P$ are in the opposite order, which introduces a sign difference). 
The
classical binary Koszul bracket is then generated by 
$\partial_P=-[d,i(P)]$, as shown by Koszul~\cite{koszul:crochet85}.

Then the MX transformation of $P$ is
$P^\sta=\frac{1}{2}P^{ab}\pi_b \pi_a$, and the corresponding $\hbar$-differential 
operator is $\hat P=\frac{1}{2}P^{ab}(-i\hbar \der{}{dx^b})(-i\hbar \der{}{dx^a}) $. 
\begin{equation}
 \Delta_P=-[d,\hat P] =
 -\hbar^2 [d,i(P)] = -\hbar^2 \partial_P \, .
\end{equation}
Then using quantum $k$-brackets formula,~\eqref{eq.clbrack}, for $k=2$, we have 
\begin{equation}
(-i \hbar )^{-2} [[\Delta_P,\o_1],\o_2](1)=  [[\partial_P,\o_1],\o_2](1) \, .        
\end{equation}

\end{remark}

\subsection{Quantum Mackenzie-Xu.}

We want the quantum descriptions in both manifestations to remain connected by the MX 
transformation, now quantum, exactly as their classical prototypes.

Introduced in~\cite{shemy:koszul}, 
the \textit{quantum Mackenzie-Xu} ($\hbar$-MX), for a vector bundle $E$, is an anti-isomorphism
(meaning  the order of factors   reversed) between the 
algebras of operators   on dual vector bundles $E$ and 
$E^*$, induced by the following pairings.

\begin{definition}
Given a volume element $\brh=\rho(x) Dx$ on the base $M$, 
for functions $f=f(x,u)\in\fun(E)$, 
$g=g(x,u^*)\in\fun(E^*)$, define
  \begin{equation}\label{eq.pairrho}
  \langle f, g \rangle_{\brh} =\int\limits_{E\times_M E^*} 
\!\!\rho(x) Dx\, Du\, Du^*\; e^{-\frac{i}{\hbar}\langle 
u,u^*\rangle} f(x,u) g(x,u^*)\,.
\end{equation}
\end{definition}

\begin{remark}
A similar pairing, but without any additional data such as $\boldsymbol{\rho}$, can be defined
for half-densities:
  \begin{equation}\label{eq.pair}
  \langle \bof, \bg \rangle =\int\limits_{E\times_M E^*}\!\!
  Dx\, Du\, Du^*\; e^{-\frac{i}{\hbar}\langle u,u^*\rangle}
f(x,u) g(x,u^*)\,.
\end{equation}
where
$\bof=f(x,u)(DxDu)^{1/2} \in\Dens_{1/2}(E)$
 and
 $\bg=g(x,u^*)(DxDu^*)^{1/2}\in\Dens_{1/2}(E^*)$.
\end{remark}

\begin{definition}
\cite{shemy:koszul}]
  The \emph{quantum Mackenzie-Xu ($\hbar$-MX) 
transformation} of an operator $A\co \fun(E_1)\to 
\fun(E_2)$ 
is defined as the adjoint  $A^{\sta}\co \fun(E_2^*)\to 
\fun(E_1^*)$
  relative to pairing~\eqref{eq.pairrho}: 
\begin{equation} \label{eq:adjointlaw}
  \langle A(f), g\rangle_{\brh}=(-1)^{\At \ft} \langle  f ,
A^{\sta}(g)\rangle_{\brh}
  \, 
\end{equation}
for operators on functions (here there is dependence on 
$\brh$). If necessary to emphasize dependence on a choice of $\brh$,
we can use notation $\sta_{\brh}$, i.e. $A^{\sta_{\brh}}$.
\end{definition}


It can also be defined  for   operators of half-densities, $A\co
\Dens_{1/2}(E_1)\to \Dens_{1/2}(E_2)$
as the adjoint  $A^{\sta}\co
\Dens_{1/2}(E_2^*)\to
\Dens_{1/2}(E_1^*)$
relative to pairing~\eqref{eq.pair}:
\begin{equation}
  \langle A(\bof), \bg\rangle=(-1)^{\At \tilde\bof} \langle
\bof , A^{\sta}(\bg)\rangle \,
\end{equation}
without dependence on any extra data. 

In what follows, we will be mostly using the version for functions~\eqref{eq:adjointlaw}.

Obviously, we have  
\begin{equation*}\label{eq.mxprod}
  (AB)^{\sta}=(-1)^{\At\Bt}B^{\sta}A^{\sta}
\end{equation*}
for all $A$ and $B$, and $(A^{\sta})^{\sta}=A$ under the 
identification $E^{**}=E$.  

It was shown~\cite{shemy:koszul} that the  quantum 
MX  transformation maps formal $\hbar$-differential 
operators to operators of the same type and that the 
principal symbol of $A^{\sta}$ is the classical MX 
transformation of the principal symbol of $A$, hence the 
name.

\begin{lemma}\label{lem:hMXbasics}
   Let $E=\Pi TM$ and $E^*=\Pi T^*M$. Then 
   \begin{align*}
   \left(f(x)\right)^{\sta}&=f(x) \, , \\
   \left(\der{}{x^a}\right)^{\sta}&=-\rho^{-1}\circ\der{}{x^a}\circ \rho \, ; \\
   (dx^a)^{\sta}&=-i\hbar\,(-1)^{\at+1}\der{}{x^*_a} \, ;\\
   \left(-i \hbar \der{}{dx^a}\right)^{\sta}&=x_a^* \, 
.
\end{align*}   
\end{lemma}
\begin{proof}
For $E=\Pi TM$, pairing~\eqref{eq.pair}  becomes 
\begin{equation*}
  \langle \o, F \rangle_{\brh} 
  =\int\limits_{\Pi TM\times_M \Pi T^*M} \!\!\rho(x) Dx\, 
D(dx)\, D(x^*)\; e^{-\frac{i}{\hbar}\,dx^a x^*_a}\; 
\o(x,dx) 
F(x,x^*)\,.
\end{equation*}
Assuming compact support of $\o$ and $F$, and 
integrating by parts, we get
\begin{align*}
  \left\langle \der{\o}{x^b}, F \right\rangle_{\brh} 
&=\int\limits_{\Pi TM\times_M \Pi T^*M} \!\!\rho(x) Dx\, 
D(dx)\, D(x^*)\; e^{-\frac{i}{\hbar}\,dx^a x^*_a}\; 
\der{\o(x,dx)}{x^b} 
F(x,x^*) \\
&=-(-1)^{\widetilde{\o}(\widetilde{b}+1)}\int\limits_{\Pi TM\times_M \Pi T^*M} \!\! 
Dx\, 
D(dx)\, D(x^*)\; e^{-\frac{i}{\hbar}\,dx^a x^*_a}\; 
\o(x,dx)
\der{(\rho(x) F(x,x^*))}{x^b}\\
&= \left\langle \o , 
-\rho^{-1}\circ\der{F(x,x^*)}{x^a}\circ \rho 
\right\rangle_{\brh} \, ,
\end{align*}
and, therefore, 
$(\der{}{x^a})^{\sta}=-\rho^{-1}\circ\der{}{x^a}\circ \rho$.

Furthermore, applying integration by parts in a slightly
different way, we get
\begin{align*}
  \left\langle  dx^b \o, F \right\rangle_{\brh} 
 &=
\int\limits_{\Pi TM\times_M \Pi T^*M} \!\!\rho(x) Dx\, 
D(dx)\, D(x^*)\; e^{-\frac{i}{\hbar}\,dx^a x^*_a}\; 
dx^b \o(x,dx) 
F(x,x^*) \\
&=\int\limits_{\Pi TM\times_M \Pi T^*M} \!\!\rho(x) Dx\, 
D(dx)\, D(x^*)\; \left(- 
\frac{\hbar}{i}\right)(-1)^{\widetilde{b}+1} 
\der{e^{-\frac{i}{\hbar}\,dx^a x^*_a}\;}{x_b^*} 
\o(x,dx)
 F(x,x^*) \\
 &=-(-1)^{(\widetilde{b}+1)\widetilde{\o}} \int\limits_{\Pi TM\times_M \Pi T^*M} 
\!\!\rho(x) Dx\, 
D(dx)\, D(x^*)\; i \hbar (-1)^{\widetilde{b}+1} 
e^{-\frac{i}{\hbar}\,dx^a x^*_a}\; 
\o(x,dx) 
\der{F(x,x^*)}{x_b^*} \\
&= (-1)^{\widetilde{dx^b} \widetilde{\o}} \left\langle \o 
, 
-i \hbar (-1)^{\widetilde{b}+1} 
(-1)^{(\widetilde{b}+1)\widetilde{\o}} 
\,\der{F(x,x^*)}{x^*_b}
 \right\rangle_{\brh} \, ,
\end{align*}
and, therefore,  
$(dx^a)^{\sta}=-i\hbar\,(-1)^{\at+1}\der{}{x^*_a}$.  

Finally,
\begin{align*}
  \left\langle  \der{\o}{dx^b}, F \right\rangle_{\brh} 
 &=
\int\limits_{\Pi TM\times_M \Pi T^*M} \!\!\rho(x) Dx\, 
D(dx)\, D(x^*)\; e^{-\frac{i}{\hbar}\,dx^a x^*_a}\; 
\der{\o}{dx^b}(x,dx) 
F(x,x^*) \\
&=-(-1)^{(\widetilde{b}+1)\widetilde{\o}} \int\limits_{\Pi TM\times_M \Pi T^*M} 
\!\!\rho(x)
Dx\, 
D(dx)\, D(x^*)\;
\o(x,dx) \der{e^{-\frac{i}{\hbar}\,dx^a x^*_a} }{dx^b}
F(x,x^*) \; \\
&=-(-1)^{(\widetilde{b}+1)\widetilde{\o}}\int\limits_{\Pi TM\times_M \Pi T^*M} 
\!\!\rho(x)
Dx\, 
D(dx)\, D(x^*)\;
\o(x,dx) \left(-\frac{i}{\hbar}\right) \, x^*_b e^{-\frac{i}{\hbar}\,dx^a x^*_a}
F(x,x^*) \; ,
\end{align*}
where we also took into account the sign in~\eqref{eq:adjointlaw}. Therefore,
$(-i\hbar \der{}{dx^a})^{\sta}=x^*_a$.
\end{proof}

\begin{lemma} 
\label{lem.mxofd}
  Let $E=\Pi TM$ and $E^*=\Pi T^*M$. Then for the de Rham 
differential $d$,
  \begin{equation}
    d^{\sta}=-i\hbar\, \delta_{\brh} \quad \text{or} \quad 
d^{\sta}=-i\hbar\, \delta\,,
  \end{equation}
  for the variant  based on $\brh$ for operators on 
functions or the variant  without $\brh$ for operators on 
half-densities, respectively. 
  Here $\delta_{\brh}$ and $\delta$ denote the divergence 
operators (or BV Laplacians) acting, respectively, on 
functions and  half-densities on $\P T^*M$. 
\end{lemma}
\begin{proof} Using previous lemma,
\begin{align*}
 d^{\sta}&=\left(\der{}{x^a}\right)^{\sta}(dx^a)^{\sta}  
    =-\rho^{-1}\circ\der{}{x^a}
         \circ \rho\circ 
(-i\hbar)\,(-1)^{\at+1}\der{}{x^*_a} 
=(-i\hbar)\,(-1)^{\at}\ 
\rho^{-1}\circ\der{}{x^a}\circ \rho\circ \der{}{x^*_a} \, ,
\end{align*}
i.e.
\begin{equation*}
  d^{\sta}=-i\hbar\,\delta_{\brh}=-i\hbar\, 
(-1)^{\at}\frac{1}{\rho(x)}\der{}{x^a}\rho(x)\der{}{x^*_a}\,
.
\end{equation*}
Similarly with the variant without $\brh$\,:
\begin{equation*}
  d^{\sta}=-i\hbar\,\delta =-i\hbar\, 
(-1)^{\at}\der{}{x^a}\,\der{}{x^*_a}\,.
\end{equation*}
\end{proof}

When starting with the usual $d$, the appearance of $i$ and 
$\hbar$ in the formula for $d^{\sta}$ might seem bothersome; 
however, when working with $\hbar$-differential operators, 
everything becomes more logical:
\begin{equation}\label{eq.danddelta}
  (-i\hbar\,d)^{\sta}= (-i\hbar)^{2}\, \delta_{\brh} \quad 
\text{or} \quad  (-i\hbar\,d)^{\sta}= (-i\hbar)^{2}\, 
\delta\,.
\end{equation}
At the left-hand side we have a first-order operator (hence 
one factor of $-i\hbar$), while at the right-hand side we 
have a second-order operator (hence two factors of 
$-i\hbar$).

\begin{remark}
The supermanifold $\Pi TM$ has an invariant Berezin volume 
element $D(x,dx)$, which is preserved by the vector field 
$d\in\Vect(\Pi TM)$; hence it is possible to identify 
densities of any weight on $\Pi TM$ together with the action 
of $d$    by  the Lie derivative just with functions on $\Pi 
TM$, i.e. pseudodifferential forms on $M$,  with the usual 
action of $d$. At the same time, $\frac{1}{2}$-densities on 
$\P T^*M$ are identified with (pseudo)multivector densities 
(which are known as  (pseudo)integral forms) on $M$. 
\end{remark}

\begin{remark}
For ordinary manifolds, formulas~\eqref{eq.danddelta}  
correspond to the classic   ``duality relation'' between   
de Rham differential and   divergence of multivector fields. 
Indeed, the pairings~\eqref{eq.pairrho},\eqref{eq.pair}, 
when specified to $E=\Pi TM$ and $E^*=\Pi T^*M$, are the 
super versions of the  ``Hodge star like'' 
duality~\cite{tv:ivb}.
\end{remark}

\begin{lemma}
  For the operator $\hat P=P(x,-i\hbar\,\der{}{dx})$ on $\fun(\Pi TM)$, we have
  \begin{equation*}
    (\hat P)^{\sta}=P\,,
  \end{equation*}
  i.e. the multiplication operator by $P=P(x,x^*)$ on $\fun(\Pi T^*M)$.
\end{lemma}
\begin{proof} 
  Indeed, from Lemma~\ref{lem:hMXbasics} we have 
$(-i\hbar\der{}{dx^a})^{\sta}=x^*_a$, and so
$\left(P(x,-i\hbar\,\der{}{dx})\right)^{\sta}=P(x,x^*)$.
\end{proof}

\begin{theorem}
\label{thm.deltaPsta}
For the operator $\Delta_P$ on $\fun(\Pi TM)$, the quantum MX  transformation maps it to the operator
\begin{equation}\label{eq.lichcorrected}
  (\Delta_P)^{\sta}=-i\hbar\,d_P-i\hbar \delta_{\brh}(P)
\end{equation}
on $\fun(\Pi T^*M)$.
\end{theorem}
\begin{proof}
  We have $(\Delta_P)^{\sta}= (-[d,\hat P])^{\sta}=[d^{\sta},(\hat 
P)^{\sta}]=-i\hbar\,[\delta_{\brh}, P]$. To calculate this commutator, 
recall that the divergence operator generates the Schouten bracket: for 
$F\in \fun(\Pi T^*M)$, we have
$    \delta_{\brh}(PF)=\delta_{\brh}(P)F + P\delta_{\brh}(F) +  \lsch  P,F\rsch\,$.
  Hence,
  \begin{equation*}
   [\delta_{\brh}, P](F) = \delta_{\brh} (PF) - P \delta_{\brh}(F) =\delta_{\brh}(P)F 
+  \lsch P,F\rsch
  \end{equation*}
 Or: 
$[\delta_{\brh}, P]=d_P+\delta_{\brh}(P)$. 
\end{proof}

With this new information about the action of the $\hbar$-MX on $\Pi T^*M$, 
we can update Fig.~\ref{dia:hcotangent} accordingly. 
\begin{figure}[ht]
 \begin{center}
\begin{tikzpicture}[thick, font=\small, node distance=1.5cm, 
                    every node/.style={align=center}]
    \node (a) {${\Pi T^*M:}$};  
    \node (b) [right=of a] {$d_P$};  
    \node (c) [right=of b] {};  
    \node (d) [right=of c] {$\HdP = \{\HSch, P\}$};  %
    \node (e) [right=of d] {$D_{d_P}=-i\hbar\,d_P-i\hbar \delta_{\brh}(P)$}; 
    \node (2e)[below=0.7cm of e] {$\Delta_P=-[d, \hat{P}]$}; 
    \node (2d)[below=0.7cm of d] {$H_P =-\{ \Hd , P(x,\pi) \}$}; 
    \node (2b) [below=0.7cm of b] {$S_\infty$: \text{higher Koszul br}};
    \node (2a) [below=0.7cm of a] {${\Pi TM:}$};  
     \draw[->
, decorate, decoration={zigzag, segment length=2mm, amplitude=0.4mm}] (d) -- (e) 
node[midway, above] {$\hbar$};
     \draw[->, decorate, decoration={zigzag, segment length=2mm, amplitude=0.4mm}] 
(2d) -- (2e) node[midway, above] {$\hbar$}; 
     \draw[<->] (d) -- (2d) node[midway, right] {MX}; \draw[<->] (e) -- (2e) 
node[midway, right] {$\hbar$-MX};
\end{tikzpicture}
\end{center}
\caption{Quantum manifestations of the $L_\infty$-algebroid $T^*M$ connected by 
$\hbar$-MX.}
\label{dia:hcotangent_withcorrectionterm}
\end{figure}

 The scalar term $i\hbar\, \delta_{\brh}(P)$ 
in~\eqref{eq.lichcorrected} gives a ``quantum 
correction'' to the Lichnerowicz differential. Obviously, it depends on a choice of a 
volume element $\brh$. A question is, if it may be possible to make this extra term 
zero altogether. This is the question about the \emph{modular class} of the 
$\Pinf$-structure, which we shall consider   in subsection~\ref{subsec.modular}.

\subsection{The commutative diagram and a preliminary  solution.} 

Recall the commutative diagram~\eqref{eq.nonclassdiagram} that holds in the $\Pinf$-case.  
 If we decorate the differentials by $-i\hbar$, we will obtain 
the relation 
\begin{equation} \label{eq:int1}
(-i\hbar d_P) \circ a_P^*=a_P^*\circ (-i\hbar d) \, . 
 \end{equation}
By applying quantum MX,
we obtain
\begin{equation*}
  (a_P^*)^{\sta}\circ d_P^{\sta} = d^{\sta}\circ (a_P^*)^{\sta}\,.
\end{equation*}
By applying Lemma~\ref{lem.mxofd}, this becomes 
\begin{equation*}
  (a^*)^{\sta}\circ (-i\hbar d_P)^{\sta} =(-\hbar^2 \delta_{\brh})\circ 
(a^*)^{\sta}\, .
\end{equation*}
That would be an intertwining relation we are looking for if we had $(-i\hbar 
d_P)^{\sta}=\Delta_P$. But we know that $(\Delta_P)^{\sta}=-i\hbar d_P-i\hbar 
\delta_{\brh}(P)$ (by Theorem~\ref{thm.deltaPsta}), so $(-i\hbar d_P)^{\sta}\neq 
\Delta_P$ unless $\delta_{\brh}(P)=0$. So, in general, we do not immediately obtain the desired intertwining operator.

But if it so happens that we can choose a volume element $\rho$ on $M$ such that 
$\delta_{\brh}(P)=0$ for our function $P\in\fun(\Pi T^*M)$, which specifies the 
$\Pinf$-structure on $M$, then we have
\begin{equation*}
  (a^*)^{\sta}\circ \Delta_P =(-\hbar^2 \delta_{\brh})\circ (a^*)^{\sta}\,,
\end{equation*}
which gives us the exact intertwining relation we are looking for.

It remains to note that the intertwining operator $(a^*)^{\sta}$ is the MX 
transformation of the pullback by the map $a\co \Pi T^*M\to \Pi TM$ (the anchor for 
the cotangent $\Linf$-algebroid). Since the ordinary
pullbacks are a particular case of quantum pullbacks, i.e. can be expressed as 
integral operators of the desired form, and by \cite[Theorem 4.2]{shemy:koszul}, the 
quantum MX transformation of $a^*$ is also has the form of a quantum pullback (the ``quantum 
dual''), so we indeed arrive at a BV-morphism in the sense of~\cite{tv:microformal}, 
between the BV manifolds $(\Pi T^*M, -\hbar^2\delta_{\brh})$ and $(\Pi TM, \Delta_P)$. 
By the results of~\cite{tv:microformal}, it will induce an $\Linf$-morphism from the 
``quantum higher Koszul brackets'' (the ``quantum brackets'' generated by $\Delta_P$) 
to the Schouten bracket. This solves our problem. Now, let us carry
 out these steps to derive an explicit integral formula.

 The following is a version of~\cite[Theorem 4.2]{shemy:koszul}.  
 \begin{theorem} \label{thm:4.2}
 Given an integral operator of the following form:
 $L\co \fun(E_2)\rightarrow 
\fun(E_1)$,
\begin{equation}\label{eq.tuda}
  f_2(x,u_2) \mapsto f_1(x,u_1)= 
  \int D u_2 \Dbar w_2 e^{\frac{i}{\hbar} \left( S(u_1;w_2)-u_2 w_2\right)} 
f_2(x,u_2) \, .
\end{equation}
Here
$u_1^i$ and $u_2^\alpha$ denote the fiber coordinates in $E_1$ and $E_2$, 
respectively, while $w_{1i}$ and $w_{2\alpha}$ denote the fiber coordinates in $E_1^*$ 
and $E_2^* $, respectively. 
 
 Then the dual of this operator is an operator of the same form:   
 $L^{\sta} \co \fun(E_1^*)\rightarrow 
\fun(E_2^*)$,
\begin{equation}
  g_1(x,w_1) \mapsto g_2(x,w_2)= 
  \int D w_1 \Dbar u_1 e^{\frac{i}{\hbar} \left( S^*(w_2;u_1)-u_1 w_1\right)} 
g_1(x,w_1) \, ,
\end{equation}
where $S^*(w_2;u_1)=S(u_1;w_2)$.
Here, for $x=(x^i)$
denoting coordinates in a supermanifold
$M$ of dimension $n|m$, we define
$\Dbar x = (2\pi \hbar)^{-n}(i \hbar)^m 
(-1)^{\frac{m(m+1)}{2}}Dx$.
\end{theorem}

Let 
\begin{align*}
 &E_1= \Pi T^* M \, , \ u_1=x_a^* \, , \
 E_1^*= \Pi T M \, , \ w_1=dx_a \,  , \\
 &E_2= \Pi T M \, , \ u_2=dy^a \, ,
 \
 E_2^*= \Pi T^* M \, , \  w_2=y^*_a \,  .
\end{align*}

Recall that anchor $a$, defined by formula~\eqref{eq:KHV_map}, 
can be considered as a thick morphism with generating 
 function~\eqref{eq:S}, $
 S = S(x^a,x^*_a; q_a,\pi_a)= x^aq_a
 +(-1)^{\at+1}\der{P}{x^*_a}(x,x^*) \pi_a$.
Its pullback of functions $a^*\co  \fun (E_2)=\fun(\Pi TM) \rightarrow 
\fun(E_1)=\fun(\Pi T^*M)$ can 
be re-written in the form of a quantum pullback or an operator of the form~\eqref{eq.tuda}
\begin{equation}
   f_1(x,x^*)= \int D(dy) \,\Dbar y^* e^{\frac{i}{\hbar} \left( 
(-1)^{\widetilde{a}+1}\der{P}{x_a^*}(x,x^*)\, y_a^* -dy^a y_a^* \right)} f_2(x,dy) \, ,
\end{equation}
where $\Dbar y^*=(2\pi \hbar)^{-m}(i \hbar)^n 
(-1)^{\frac{n(n+1)}{2}}Dy^*$

By Thm.~\ref{thm:4.2}, operator $(a^*)^\sta \co \fun(E_1^*)=\fun(\Pi TM) \rightarrow 
\fun(E_2^*) =\fun(\Pi 
T^*M)$ can be written as follows:
\begin{equation}
  g_2(x,y^*)= \int D(dx) \,\Dbar x^* e^{\frac{i}{\hbar} \left( 
(-1)^{\widetilde{a}+1}\der{P}{x_a^*}(x,x^*)y_a^* -dx^a x_a^* \right) } g_1(x,dx) \, ,
\end{equation}
where $\Dbar x^*=(2\pi \hbar)^{-m}(i \hbar)^n 
(-1)^{\frac{n(n+1)}{2}}Dx^*$.

We summarize our results in the following theorem.
\begin{theorem} Given a $P_\infty$
structure on a supermanifold $M$,
specified by an even $P=P(x,x^*) \in \fun(T^*M)$, $\lsch P,P\rsch=0$.
Consider the $\hbar$-differential operator
\begin{equation}
\Delta_P=-[d,\hat P] \, , 
\end{equation}
where 
$\hat P(x,x^*) = P(x,-i \hbar \der{}{dx})$, which generates higher Koszul brackets
on the algebra of differential forms $\fun(\Pi TM)$. Consider the $\hbar$-modification 
of the divergence operator,  
\begin{equation*}
  -\hbar^2\,\delta_{\brh}=-\hbar^2\, 
(-1)^{\at}\frac{1}{\rho(x)}\der{}{x^a}\rho(x)\der{}{x^*_a}\,
,
\end{equation*}
which generates the Schouten bracket on the algebra of multivector fields 
$\fun(\Pi T^* M)$. Suppose  $\delta_{\brh}(P)=0$. 

Then $\Delta_P$ and $-\hbar^2\,\delta_{\brh}$ are intertwined by $I=(a^*)^{\sta}$, 
$(a^*)^{\sta}\circ \Delta_P =(-\hbar^2 \delta_{\brh})\circ (a^*)^{\sta}$\,, where 
\begin{equation*}
  (a^*)^\sta \co \fun(\Pi TM) \rightarrow \fun(\Pi  T^*M)
\end{equation*}
can be written as follows:
\begin{equation}
  g_2(x,y^*)= \int D(dx) \,\Dbar x^* e^{\frac{i}{\hbar} \left( 
(-1)^{\widetilde{a}+1}\der{P}{x_a^*}(x,x^*)y_a^* -dx^a x_a^* \right) } g_1(x,dx) \, ,
\end{equation}
where $\Dbar x^*=(2\pi \hbar)^{-m}(i \hbar)^n 
(-1)^{\frac{n(n+1)}{2}}Dx^*$.
\end{theorem}

It remains to analyse what will happen if  $\delta_{\brh}(P)=0$ is \emph{not} satisfied. 

\subsection{Conditions for finding $\rho$ such that $\delta_{\rho}(P)=0$ and 
what to do when this is not possible; the modular class of a 
$\Pinf$-structure.} \label{subsec.modular}
We already met the extra 
term $\delta_{\brh}(P)$  in~\cite{shemy:koszul} as a quantum correction necessary for a definition of a BV operator generating the cotangent $\Linf$-\textbf{bi}algebroid in the manifestation on $\Pi T^*M$. It was noted there that it can be seen as   representing the  modular class  of a $\Pinf$-structure on $M$ specified by $P\in\fun(\Pi T^*M)$\,---\,that is, the cohomology class   generalizing   Weinstein's modular class of a Poisson structure~\cite{weinstein:modular} and similar to the modular classes of Lie algebroids~\cite{evens-lu-weinstein:1999} and $Q$-manifolds~\cite{lyakhovich:mosman1}\cite{tv:qman-esi}. 

Indeed, the following proposition holds:

\begin{proposition}
\label{prop.modu}
  The function $\delta_{\brh}(P)\in \fun(\Pi T^*M)$ is annihilated by $d_P$. If a volume element $\brh$ on $M$ is replaced by $\brh'=e^f\brh$, for some $f\in\fun(M)$, then $\delta_{\brh}(P)$ will be replaced by $\delta_{\brh'}(P)=\delta_{\brh}(P)\pm d_P(f)$. 
\end{proposition}
Therefore $\delta_{\brh}(P)$ defines a cohomology class $\mu_P=[\delta_{\brh}(P)]\in H^{\text{odd}}\bigl(\fun(\Pi T^*M), d_P\bigr)$ that depends only on $P$ and not on a choice of $\brh$. We call it, \emph{the modular class of a $\Pinf$-structure}.   Note that unlike the usual Poisson case, this cohomology is only $\Z$-graded. 

Another difference with the ordinary Poisson case is that vanishing of the class $\mu_P$ for a $\Pinf$-structure $P$ does not guarantee the existence of an  invariant measure  on $M$, i.e.   a volume element $\brh\in \Vol(M)$ such that $\delta_{\brh}(P)=0$. 

(In the ordinary Poisson case, $P$ is quadratic and $\delta_{\brh}(P)$ has degree $1$. It corresponds to a Poisson vector field on $M$, Weinstein's modular vector field;   adding   $d_P(f)=\lsch P,f\rsch$ corresponds to adding to it the Hamiltonian vector field $\{f,-\}=\lsch\lsch P,f\rsch,-\rsch$. So vanishing of this class in Poisson cohomology is equivalent to the possibility of choosing a volume element $\brh'$ on $M$ such that $\delta_{\brh'}(P)$ vanishes. In the $\Pinf$ case, a similar argument would give only a volume element on $\Pi T^*M$, but which does not necessarily arise from a $\brh'$ on $M$.)

\begin{proof}[Proof of Proposition~\ref{prop.modu}]
  We have $\lsch P, P\rsch=0$. The operator $\d_{\brh}$  is a  derivation of the Schouten bracket (as a generating operator). Hence we have  
  \begin{equation*}
    0=\d_{\brh}(\lsch P, P\rsch)= -\lsch \d_{\brh}(P), P\rsch - \lsch P, \d_{\brh}(P)\rsch= -2\lsch P, \d_{\brh}(P)\rsch =-2d_P\bigl(\d_{\brh}(P)\bigr)\,,
  \end{equation*}
 so  $d_P\bigl(\d_{\brh}(P)\bigr) =0$. Also, if $\brh$ is  replaced by $\brh'=e^f\brh$, then $\d_{\brh'}=\d_{\brh}\pm \lsch f,-\rsch$ (easiest to see from the explicit formula for $\d_{\brh}$). Hence, $\d_{\brh'}(P)=\d_{\brh}(P)\pm \lsch f,P\rsch=\d_{\brh}(P)\pm d_P(f)$.
\end{proof}

For better clarfication, we can connect the construction of $\mu_P$ with the general notion of the modular class of a $Q$-manifold (due independently to~\cite{lyakhovich:mosman1} and~\cite{tv:qman-esi}). If $N$ is a $Q$-manifold, the class $\mu_Q\in H^{\text{odd}}(\fun(N),Q)$ is defined as  the class $\mu_P:=[\div_{\hat\brh}(Q)]$ of the divergence of the field $Q$ with respect to some volume element $\hat\brh\in \Vol(N)$; i.e. $\div_{\hat\brh}(Q)=(\hat\brh)^{-1}L_Q(\hat\brh)$, which may be symbolically written as $\div_{\hat\brh}(Q)=L_Q(\ln\hat\brh)$. Then replacing $\hat\brh$ by $\hat\brh'=e^F\hat\brh$ will result in replacing $\div_{\hat\brh}(Q)$ by $\div_{\hat\brh}(Q)+Q(F)$. 

If $M$ is a $\Pinf$-manifold, we can consider a $Q$-manifold $N=\Pi T^*M$ with $Q=d_P$. Then $\hat\brh\in \Vol(\Pi T^*M)$ and a special case is $\hat\brh=\brh^2$ (as follows from transformation of coordinates on $\Pi T^*M$). In this case, 
\begin{equation*}
  \div_{\brh^2}(Q)=2 (\brh)^{-1}L_{\lsch P,-\rsch}(\brh)\,.
\end{equation*}
Recall that the BV operator on functions on an odd Poisson manifold equipped with a volume form is given by $\Delta_{\hat\brh}(F)=\div_{\brh}(X_F)$~\cite{hov:deltabest}, \cite{schwarz:bv}, where $X_F$ is the odd Hamiltonian vector field corresponding to a function $F$. In particular, on $\Pi T^*M$, we have
\begin{equation*}
  \Delta_{\hat\brh}(F)=\div_{\hat\brh}(\lsch F, -\rsch)
\end{equation*}
and in the special case of $\hat\brh=\brh^2$ we have $\Delta_{\brh^2}=2\delta_{\brh}$\footnote{The presence or absence of extra factors such as $2$ or $1/2$ depends on convention.}. Therefore we have $\mu_Q=2\mu_P$ for  $Q=d_P$. 

Thus vanishing of $\mu_P$ is equivalent to the existence of a $d_P$-invariant volume element on $\Pi T^*M$ (which not necessarily comes from  $M$), in particular the possibility of replacing $\delta_{\brh}$ by another BV operator $\frac{1}{2}\Delta_{e^{2F}\brh^2}=\delta_{\brh}+d_P(F)$. Note that the modified operator will be odd, but not in general of any $\ZZ$-degree, unlike the classical divergence operator $\delta_{\brh}$.


%

\subsection{Completion of the construction.} \label{subsec.complet}

Suppose $\mu_P=0$. That means that there is some $F\in\fun(\Pi T^*M)$ such that $\delta_{\brh}(P)=d_P(F)$. Consider again the commutative diagram
\begin{equation}\label{eq.diagram3}
    \begin{CD} \Mult  (M)@>{-i\hbar d_P}>> \Mult(M)\\
                @A{a^*}AA         @AA{a^*}A\\
                \O(M)@>{-i\hbar d}>> \O(M) \, 
    \end{CD}
\end{equation}
and recall our previous discussion.
As we know that $(\Delta_P)^{\sta}=i\hbar d_P+i\hbar \delta_{\brh}(P)$, we obtain that $(\Delta_P)^{\sta}=i\hbar d_P+i\hbar d_P(F)$. If we simply replace $d_P$ by $d_P+ d_P(F)$ in the diagram, it will cease to commute. But we can modify the vertical arrows and restore the commutativity.

\begin{lemma}
Set   $a_{\text{quant}}^*:=e^{-F}a^*$.
Then the diagram
\begin{equation}\label{eq.diagram4}
    \begin{CD} \Mult  (M)@>{-i\hbar d_P-i\hbar d_P(F)}>> \Mult(M)\\
                @A{a_{\text{quant}}^*}AA         @AA{a_{\text{quant}}^*}A\\
                \O(M)@>{-i\hbar d}>> \O(M) \, 
    \end{CD}
\end{equation}
(a modification of diagram~\eqref{eq.diagram3})
will be commutative. 
\end{lemma}
\begin{proof}
  Indeed, for $\o\in \O(M)$, we have 
  \begin{multline*}
    (-i\hbar d_P-i\hbar d_P(F))(e^{-F}a^*(\o))=-i\hbar d_P(e^{-F}a^*(\o)) -i\hbar d_P(F)e^{-F}a^*(\o) =\\
    -i\hbar(-d_P(F))e^{-F}a^*(\o)  -i\hbar e^{-F} d_P(a^*(\o)) -i\hbar d_P(F)e^{-F}a^*(\o)=
     -i\hbar e^{-F} d_P(a^*(\o))= \\
     -i\hbar e^{-F} a^*(d(\o))= 
      e^{-F} a^*(-i\hbar d(\o))\,,
  \end{multline*}
  due to the commutativity of~\eqref{eq.diagram3}. (Or, quicker: the top arrow is $-i\hbar d_P-i\hbar d_P(F)=e^{-F}(-i\hbar d_P)e^{F}$, and then the statement is immediate.)
\end{proof}

Therefore, the sought-after operator intertwining $-\hbar^2\delta_{\brh}$ and $\Delta_P$ will be 
\begin{equation*}
  \bigl(a_{\text{quant}}^*\bigr)^{\sta}= \bigl(e^{-F}a^*\bigr)^{\sta}\,.
\end{equation*}

It remains to write down explicitly both $a_{\text{quant}}^*$ and $\bigl(a_{\text{quant}}^*\bigr)^{\sta}$ as integral operators and see that they are indeed quantum pullbacks that are lifts of the classical pullbacks $a^*$ and $(a^{\du})^*$.


\end{document}